\documentclass[12pt]{iopart}

\usepackage{iopams}

\expandafter\let\csname equation*\endcsname\relax

\expandafter\let\csname endequation*\endcsname\relax

\usepackage{amsmath, amsthm, amssymb}
\usepackage{graphicx}
\usepackage{tensor}

\usepackage[dvipsnames]{xcolor}
\usepackage{tikz} % rysunki
\usepackage{wasysym} %symbole astronomiczne 
\usepackage{microtype} %dzieleni słów
\usepackage[toc,page]{appendix}
\newtheorem{theorem}{Theorem}[section]

\newtheorem{lemma}[theorem]{Lemma}

\begin{document}

\title[Revisiting timelike and null geodesics in the Schwarzschild spacetime]{Revisiting timelike and null geodesics in the Schwarzschild spacetime: general expressions in terms of Weierstrass elliptic functions}

\author{Adam Cie\'{s}lik and Patryk Mach}

\address{Instytut Fizyki Teoretycznej, Uniwersytet Jagiello\'{n}ski, \L{}ojasiewicza 11, 30-348 Krak\'{o}w, Poland}
 \eads{\mailto{adam.cieslik@doctoral.uj.edu.pl}, \mailto{ patryk.mach@uj.edu.pl}}

\begin{abstract}
    The theory of Schwarzschild geodesics is revisited. Basing on a result by Weierstrass and Biermann, we derive a formula describing all non radial, timelike and null trajectories in terms of Weierstrass elliptic functions. Quite remarkably, a single formula works for an entire geodesic trajectory, even if it passes through turning points. Using this formula, we derive expressions for the proper and coordinate time along the geodesic.
\end{abstract}


\section{Introduction}
In this paper we revisit some elements of the theory of geodesics in the Schwarzschild spacetime. The motivation for repeating this classic calculation in a new form stems from the works on the kinetic description of the Vlasov gas on the Schwarzschild spacetime and the accretion of the Vlasov gas onto Schwarzschild black holes \cite{rioseco_accretion_2017, rioseco_spherical_2017, mach_accretion_2021, mach_accretion_2021-1,mach_acta,Olivier_disks,gamboa_2021,cieslik_2020}. In particular, we were motivated by an ongoing project aiming at constructing Monte Carlo type simulations of the gas consisting of collisionless particles moving around the Schwarzschild black hole. Having those applications in mind, we put special emphasis on unbounded trajectories---they are neglected in many discussions, but they play a crucial role in the description of Bondi-type accretion processes.

Existing descriptions of Schwarzschild geodesics differ in the parametrization and also in types elliptic functions used to express the solutions. As usual, different prescriptions appear to be more or less convenient, depending on the actual problem at hand. Our goal was to specify the constants of motion (in particular the energy and the angular momentum) together with the initial location of the particle, and obtain the corresponding trajectory in an exact and reliable manner. We achieve this aim using Weierstrass elliptic functions, but our prescription turns out to be different than existing ones (or at least the ones we are aware of). The main result presented in this paper is a concise formula describing all types of timelike and null trajectories in the Schwarzschild metric (except for the purely radial ones), based on a theorem due to Biermann and Weierstrass. 

The problem of an analytic description of the motion of test particles in the Schwarzschild spacetime is nearly as old as the Schwarzschild solution itself. The first attempt to solve geodesic equations in the Schwarzschild spacetime was published in 1917 by Droste, together with a derivation of the Schwarzschild metric \cite{droste_field_1917} (this paper is now also available as a ``Golden Oldie'' reprint \cite{droste_reprint}). Droste expressed his solution in terms of the Weierstrass elliptic function. Thirteen years later Hagihara  gave a full description of the motion of test particles around a Schwarzschild black hole \cite{hagihara_theory_1930}. His work contains a complete characterization of all types of allowed orbits and is now a classic position in the theory of Schwarzschild geodesics.

Simultaneously to the development of solutions based on Weierstrass functions, Forsyth, Greenhill, Morton, Darwin, Mielnik, and Pleba\'{n}ski succeeded in expressing Schwarzschild geodesics in terms of Jacobi elliptic functions and Legendre integrals \cite{forsyth_note_1920, greenhill_newton-einstein_1921, morton_lxi_1921, darwin_gravity_1959, darwin_gravity_1961, mielnik_study_1962}. Publication years of these papers span across several decades. As a historical remark, let us note that the authors of Refs.\ \cite{darwin_gravity_1959, mielnik_study_1962}, published in 1959 and 1962, already admitted that their calculations could had been made almost forty years earlier, as all required mathematical tools were already known at that time.

In subsequent years, researchers were mostly aware of the existence of two general ways of solving Schwarzschild geodesic equations, choosing between them according to their preferences and needs \cite{synge_relativity_1960, metzner_observable_1963, sharp_1979, chandrasekhar_mathematical_1983, gackstatter_uber_1983, rodriguez_orbits_1987, dabrowski_1995, slezakova_geodesic_2006, hioe_exact_2009, hioe_characterizing_2010, grunau_2011, gibbons_application_2012, kostic_analytical_2012, hackmann_2014, munoz_orbits_2014,semerak_2015, rosales_vera_2018, kostic_cadez_2005, luminet_1979}. Similar methods were also applied to an analysis of the geodesic motion in other spherically symmetric spacetimes, for which geodesic equations are solvable \cite{cruz_geodesic_2005, soroushfar_analytical_2015, chatterjee_analytic_2019, panotopoulos_orbits_2021}. In some of such cases equations of motion contain effectively a fifth degree polynomial expression, leading to hyperelliptic Abelian integrals \cite{kraniotis_compact_2003, hackmann_geodesic_2008, soroushfar_analytical_2015}.

A slightly different approach to the geodesic motion in the Schwarzschild metric was recently proposed by Scharf \cite{scharf_schwarzschild_2011}. Scharf's analysis is based on a simplified version of a result which we will refer to as the Biermann-Weierstrass formula.

According to Refs.\ \cite{greenhill_applications_1892, whittaker_course_1927, reynolds_exact_1989}, around 1860 Weierstrass obtained a general solution of an ordinary differential equation
\begin{equation}
\label{weierstrass_diff_eq}
    \frac{d^2 x(t)}{dt^2} = \alpha x(t)^3 + \beta x(t)^2 + \gamma x(t) + \delta,
\end{equation}
where $\alpha$, $\beta$, $\gamma$, and $\delta$ are constant coefficients, which is roughly equivalent to an integral problem
\begin{equation}
t = \int_{x_0}^{x(t)} \frac{dx^\prime}{\sqrt{f(x^\prime)}},    
\end{equation}
where $f$ is a quartic polynomial, and $x_0 = x(0)$. Weierstrass' solution was later published by his student Biermann, as a part of his inaugural dissertation \cite{biermann_problemata_1865}. The Weierstrass-Biermann formula is derived in Greenhill's textbook \cite{greenhill_applications_1892}; it appears also in the classic textbook by Whittaker and Watson \cite{whittaker_course_1927}. On the other hand, it is absent in other textbooks dealing with elliptic functions from that time \cite{cayley_elementary_1876, baker_elliptic_1890, hardy_integration_1916, hancock_elliptic_1917, forsyth_theory_1918}, nor does it appear in standard textbooks and tables used by physicists \cite{bateman_higher_1953, abramowitz_handbook_1964, byrd_handbook_1971, chandrasekharan_elliptic_1985, akhiezer_elements_1990, gradshtein_table_2007}.

The Biermann-Weierstrass formula for $x = x(t)$ simplifies, if $x_0$ is a zero of the polynomial $f$, and this version is used by Scharf. In the context of the geodesic motion this requirement restricts the choice of the starting (reference) point of the trajectory to turning points. In our work, we start with the general Biermann-Weierstrass expression, and hence this restriction is removed. A single formula [Eq.\ (\ref{xi_psi}) of this article] describes all timelike and null geodesic trajectories, except purely radial ones. Moreover, it is valid along the entire trajectory, even if it passes through turning points. The latter fact is not immediately obvious from the derivation of Eq.\ (\ref{xi_psi}), but it can be shown using addition theorems for elliptic functions.

The affine parameter and the coordinate time along a geodesic can be obtained as integrals involving the Biermann-Weierstrass expression. While, in principle, they can be evaluated assuming the general form of Eq.\ (\ref{xi_psi}), the resulting formulas are lengthy and thus of little practical use. For completeness, we decided to provide appropriate expressions for the affine parameter and the coordinate time, assuming the simplified version of the Biermann-Weierstrass formula.

A proof of the Biermann-Weierstrass formula is given in  \ref{appendix:BW_theorem}; we decided to provide this material, since existing, known to us proofs of the Biermann-Weierstrass formula are rather hard to follow in detail. Apart from a sketch of the proof given in Biermann's dissertation \cite{biermann_problemata_1865}, there is a proof in Greenhill's textbook \cite{greenhill_applications_1892}, and its more modern account in \cite{reynolds_exact_1989}. We fill some minor gaps missing in \cite{reynolds_exact_1989}. 

The order of this paper is as follows. The next section specifies horizon-penetrating coordinates used in this paper. Equations of motion are derived in Sec.\ \ref{sec:geodesicmotion}. The main result of this paper---a single formula describing non-radial, timelike and null trajectories---is given in Sec.\ \ref{sec:solution}. In Section \ref{sec:rangepsi} we discuss the range of the true anomaly parameter. The proper time and the coordinate time along a geodesic are computed in Sec.\ \ref{sec:propertime} and \ref{sec:coordinatetime}, respectively. The proof of the Biermann-Weierstrass formula is given in \ref{appendix:BW_theorem}. \ref{sec:class_of_traj} contains a brief classification of timelike and null geodesics. In \ref{appendix:psi_xi} we evaluate certain elliptic integrals, required to control the range of the true anomaly parameter for unbounded orbits.

Throughout the paper we use geometric units with $c = G = 1$, where $c$ is the speed of light, and $G$ denotes the gravitational constant. The signature of the metric is assumed to be $(-,+,+,+)$. Spacetime dimensions are labeled with Greek indices, $\mu = 0,1,2,3$.

%=================================================================================================================================
%                                               Horizon-penetrating coordinates
%=================================================================================================================================

\section{Horizon-penetrating coordinates}

We will work in spherical coordinates $(t,r,\theta,\varphi)$. In its simplest form (in the so-called Schwarzschild coordinates) the Schwarzschild metric can be written as
\begin{equation}
g = -N d\bar t^2 + \frac{d \bar r^2}{N} + \bar r^2 d \theta^2 + \bar r^2 \sin^2 \theta d \varphi^2,
\end{equation}
where
\begin{equation}
N = 1 - \frac{2M}{\bar r},
\end{equation}
and $M$ denotes the black hole mass. Since ultimately we envisage applications of geodesic solutions to accretion problems, we will also need coordinates in which the metric is explicitly regular at the horizon. Such coordinates can be easily obtained by a choice of the time foliation. The transformation
\begin{equation}
t = \bar t + \int^{\bar r} \left[ \frac{1}{N(s)} - \eta(s) \right] ds, \quad r = \bar r,
\end{equation}
where $\eta = \eta(\bar r)$ is a function of radius $\bar r$, yields the metric in the form
\begin{equation}
\label{schwbl}
g = -N dt^2 + 2 (1 - N \eta) dt dr + \eta ( 2 - N \eta) dr^2 + r^2 d \theta^2 + r^2 \sin^2 \theta d \varphi^2.
\end{equation}
The function $\eta$ defines the time foliation. A popular choice $\eta \equiv 1$ leads to coordinates which are manifestly regular at the horizon $r = 2M$, and which are sometimes referred to as Eddington-Finkelstein coordinates. Note that since we only change the time foliation, the radial coordinate $r$ retains its interpretation as the areal radius.

Contravariant components of the metric (\ref{schwbl}) are given by
\begin{equation}
g^{tt} = \eta (-2 + N \eta), \quad g^{tr} = 1 - N \eta, \quad g^{rr} = N, \quad g^{\theta \theta} = \frac{1}{r^2}, \quad g^{\varphi \varphi} = \frac{1}{r^2 \sin^2 \theta}.
\end{equation}
Moreover,
\begin{equation}
\left( g^{tr} \right)^2 - g^{rr} g^{tt} = 1.
\end{equation}

%=================================================================================================================================
%                                               Geodesic motion 
%=================================================================================================================================

\section{Geodesic motion}
\label{sec:geodesicmotion}

There are many well-known descriptions of the geodesic motion. In order to keep a connection with the works \cite{rioseco_accretion_2017, rioseco_spherical_2017, mach_accretion_2021, mach_accretion_2021-1}, we will work in the Hamiltonian framework. The Hamiltonian $H$ describing the geodesic motion of a free particle can be chosen as
\begin{equation}
    H = \frac{1}{2} g^{\mu\nu}(x^\alpha) p_{\mu} p_{\nu}.
\end{equation}
Here $(x^\mu, p_\mu)$ are treated as canonical variables, and $H$ depends on $x^\alpha$ through $g^{\mu \nu}(x^\alpha)$. It is easy to verify that the Hamilton equations
\begin{equation}
\label{hamilton_eqs}
\frac{dx^\mu}{d\tilde{s}} = \frac{\partial H}{\partial p_\mu}, \quad \frac{dp_\mu}{d\tilde{s}} = - \frac{\partial H}{\partial x^\mu}
\end{equation}
lead to standard geodesic equations of the form
\begin{equation}
\frac{d^2 x^\mu}{d \tilde s^2} + \Gamma^\mu_{\alpha \beta} \frac{d x^\alpha}{d \tilde s} \frac{d x^\beta}{d \tilde s} = 0.
\end{equation}
For timelike geodesics we choose the affine parameter $\tilde{s}$ as a rescaled proper time $\tilde \tau$, i.e., $\tilde{s}=\tilde{\tau}/m$, where $m$ is the particle rest mass. The four velocity $u^\mu = d x^\mu/ d\tilde{\tau}$ is normalized as $g_{\mu \nu}u^\mu u^\nu = -1$. We require that $p^\mu = d x^\mu/d\tilde{s}$, and that $H = \frac{1}{2} g^{\mu \nu} p_\mu p_\nu = - \frac{1}{2} m^2$.

For null geodesics $m = 0$ and $H = \frac{1}{2}g^{\mu \nu} p_\mu p_\nu = 0$. In this case the affine parameter $\tilde s$ is unique up to a transformation
\begin{equation}
    \tilde s \to \tilde s^\prime = \alpha \tilde s + \beta,
\end{equation}
with $\alpha > 0$ and an arbitrary $\beta$. Such an affine reparametrization implies a rescaling of the tangent vector
\begin{equation}
 p^\mu \to {p^\prime}^\mu = \frac{1}{\alpha} p^\mu.  
\end{equation}

The explicit form of the Hamiltonian $H$, assuming metric (\ref{schwbl}), reads
\begin{equation}
H = \frac{1}{2} \left[ g^{tt}(r)p_t^2 + 2 g^{tr}(r) p_t p_r + g^{rr}(r) p_r^2 + \frac{1}{r^2} \left( p_\theta^2 + \frac{p_\varphi^2}{\sin^2 \theta} \right) \right].
\end{equation}
Since $H$ depends neither on $t$ nor on $\varphi$, the momentum components $E \equiv - p_t$ (the energy) and $l_z \equiv p_\varphi$ are constants of motion. The Hamiltonian $H$ is also independent of $\tilde s$, and hence it is also conserved. A simple calculation allows one to check that the total angular momentum
\begin{equation}
\label{constraintl}
l = \sqrt{p_\theta^2 + \frac{p_\varphi^2}{\sin^2 \theta}}
\end{equation}
is another constant of motion.

The remaining momentum components $p_\theta$ and $p_r$ can be expressed as
\begin{equation}
\label{ptheta}
p_\theta = \epsilon_\theta \sqrt{l^2 - \frac{l_z^2}{\sin^2 \theta}}
\end{equation}
and
\begin{equation}
\label{pr}
p_r = \frac{g^{tr} E + \epsilon_r \sqrt{\left[ (g^{tr})^2 - g^{tt} g^{rr} \right] E^2 - g^{rr} \left( m^2 + \frac{l^2}{r^2} \right)}}{g^{rr}} = \frac{(1 - N \eta)E + \epsilon_r \sqrt{E^2 - \tilde U_{l,m}(r)}}{N},
%p_r = \epsilon_r \left(1 - \frac{r_\mathrm{s}}{r} \right)^{-1} \sqrt{E^2 - \tilde U_{l,m}(r)},
\end{equation}
where
\begin{equation}
\tilde U_{l,m}(r) = \left(1 - \frac{2M}{r} \right) \left( m^2 + \frac{l^2}{r^2} \right)
\end{equation}
is the radial effective potential, and where we have introduced the signs $\epsilon_\theta = \pm 1$, and $\epsilon_r = \pm 1$, corresponding to the directions of motion. Here and in what follows, the symbol $\sqrt{}$ denotes the non-negative branch of the square root. Equation (\ref{ptheta}) follows directly from Eq.\ (\ref{constraintl}). Equation (\ref{pr}) can be obtained from the equality $H = - \frac{1}{2}m^2$. Note also that, contrary to the formula for $p_r$, the expression for $p^r$,
\begin{equation}
p^r = \epsilon_r \sqrt{E^2 - \tilde U_{l,m}(r)}, 
\end{equation}
does not depend on $\eta$. On the other hand
\begin{equation}
p^t =  \frac{E}{N} + \epsilon_r \frac{1 - N \eta}{N}\sqrt{E^2 - \tilde U_{l,m}(r)},
\end{equation}
and this expression enters the equations of motion listed below. Also note that, while for $\epsilon_r = +1$ (outgoing motion) the expression for $p_r$ diverges at the horizon, the corresponding expression for $\epsilon_r = -1$ is perfectly regular (particles are allowed to fall into the black hole). An expression for $p_r$ with $\epsilon_r = -1$, manifestly regular at $r = 2M$, reads
\begin{equation}
p_r = - \eta E + \frac{m^2 + \frac{l^2}{r^2}}{E + \sqrt{E^2 - \tilde U_{l,m}(r)}}.
\end{equation}

Equations $dx^\mu/d \tilde s = \partial H/\partial p_\mu$ can be written as
\begin{subequations}
\label{eqsofmotion1}
\begin{alignat}{1}
	\frac{dr}{d\tilde{s}} = \frac{\partial H}{\partial p_{r}} &= \epsilon_r \sqrt{E^2 - \tilde U_{l,m}(r)},\\
    \frac{d\varphi}{d\tilde{s}} = \frac{\partial H}{\partial p_{\varphi}} &=  \frac{l_{z}}{r^{2}\sin^2 \theta}, \\
    \frac{d\theta}{d\tilde{s}} = \frac{\partial H}{\partial p_{\theta}} &= \frac{\epsilon_\theta}{r^{2}} \sqrt{ l^{2} -\frac{l_{z}^2}{\sin^2 \theta}},\\
    \frac{dt}{d\tilde{s}} = \frac{\partial H}{\partial p_{t}} &=  \frac{E}{N} + \epsilon_r \frac{1 - N \eta}{N}\sqrt{E^2 - \tilde U_{l,m}(r)}. \label{dtdstilde}
\end{alignat}
\end{subequations}
Note that the dependence on $\eta$ appears only in Eq.\ (\ref{dtdstilde}). In standard Schwarzschild coordinates $N \eta \equiv 1$ and $dt/d\tilde s = E/N$.

It is convenient to work in dimensionless rescaled variables. For timelike geodesics we define such variables as in \cite{rioseco_accretion_2017}, i.e., by
\begin{alignat}{7}
    t = M \tau, &\; r = M \xi, &\; p_r =  m \pi_\xi, &\; p_\theta  =  M m \pi_\theta, &\; E  =  m \varepsilon, &\; l =  M m \lambda, &\; l_z = M m \lambda_z.
\end{alignat}
In addition, a new affine parameter $s$ is defined by
\begin{equation}
    \tilde{s} = \frac{M}{m} s.
\end{equation}
For null geodesics $m = 0$. We introduce an arbitrary mass parameter $\tilde m >0$, and define
\begin{alignat}{7}
    t = M \tau, &\; r = M \xi, &\; p_r =  \tilde m \pi_\xi, &\; p_\theta  =  M \tilde m \pi_\theta, &\; E  =  \tilde m \varepsilon, &\; l =  M \tilde m \lambda, &\; l_z = M \tilde m \lambda_z,
\end{alignat}
and $\tilde s = (M/\tilde m) s$.

In terms of these dimensionless variables, the equations of motion (\ref{eqsofmotion1}) can be written as
\begin{subequations}
\label{eqsofmotion3}
\begin{eqnarray}
    \frac{d\xi}{ds} & = &  \epsilon_r \sqrt{\varepsilon^{2} - U_\lambda(\xi)}\label{xi_mot},\\
    \frac{d\varphi}{ds} & = & \frac{\lambda_{z}}{\xi^{2}\sin^2 \theta}\label{phi_mot},\\ 
    \frac{d\theta}{ds} & = & \epsilon_\theta \frac{1}{\xi^{2}} \sqrt{ \lambda^{2} - \frac{\lambda_{z}^2}{\sin^2 \theta}}\label{theta_mot},\\
    \frac{d\tau}{ds} & = & \frac{\varepsilon}{N(\xi)} + \epsilon_r \frac{1 - N(\xi) \eta(\xi)}{N(\xi)}\sqrt{\varepsilon^2 -  U_\lambda(\xi)} \label{tau_mot},
\end{eqnarray}
\end{subequations}
where $N(\xi) = 1 - 2/\xi$. The dimensionless radial potential reads
\begin{equation}
\label{eff_pot}
U_\lambda(\xi) = \left( 1 - \frac{2}{\xi} \right)\left(  1 + \frac{\lambda^2}{\xi^2}\right) = 1 - \frac{2}{\xi} + \frac{\lambda^{2}}{\xi^{2}} - \frac{2\lambda^{2}}{\xi^{3}}
\end{equation}
for timelike geodesics, and
\begin{equation}
\label{eff_pot_null}
U_\lambda(\xi) = \left( 1 - \frac{2}{\xi} \right) \frac{\lambda^2}{\xi^2}
\end{equation}
for null ones.

It is well known that geodesic motion in the Schwarzschild spacetime is confined to a plane. Choosing the coordinate system so that $\theta \equiv \pi/2$, and $d \theta/d s \equiv 0$, we get $\lambda^2 = \lambda_z^2$, and thus $\lambda_z = \pm \lambda$. We will adopt a convention with $\lambda \ge 0$ and define the angle in the orbital plane (the so-called true anomaly) $\psi = \mathrm{sgn} (\lambda_z) \varphi$. The relevant equations of motion can be written as
\begin{subequations}
\label{eqsofmotion4}
\begin{eqnarray}
    \frac{d\xi}{ds} & = & \epsilon_r \sqrt{\varepsilon^{2} - U_\lambda(\xi)}, \label{xi2_mot}\\
    \frac{d\psi}{ds} & = & \frac{\lambda}{\xi^{2}}, \label{psi2_mot}\\ 
    \frac{d\tau}{ds} & = & \frac{\varepsilon}{N(\xi)} + \epsilon_r \frac{1 - N(\xi) \eta(\xi)}{N(\xi)}\sqrt{\varepsilon^2 -  U_\lambda(\xi)} \label{tau2_mot}. 
\end{eqnarray}
\end{subequations}

System (\ref{eqsofmotion4}) can also be obtained by introducing standard orbital elements such as the orbital inclination, the argument of periapsis, the argument of latitude, and the true anomaly (see, e.g., \cite{kostic_analytical_2012}). Another possibility to (partially) decouple the equations of motion (\ref{eqsofmotion3}) is to introduce the so-called Mino time \cite{mino_perturbative_2003}.

A qualitative analysis of the effective radial potential allows for a general classification of different types of orbits. This is done, to some extent, in \ref{sec:class_of_traj}, both for timelike and null orbits. In general, we divide trajectories into bound and unbound ones. Undbound trajectories can either start at infinity and plunge into the black hole (we refer to such trajectories as absorbed ones). The second large class of unbound trajectories consists of orbits characterized by sufficiently large angular momentum. In this case the particles arriving from infinity are scattered off the centrifugal barrier (these trajectories are referred to as scattered ones).

%=================================================================================================================================
%                                               Solution of equations of motion 
%=================================================================================================================================

\section{Solution of equations of motion}
\label{sec:solution}

\subsection{Timelike geodesics}

We will start our analysis with timelike geodesics. Given the form of Eqs.\ (\ref{eqsofmotion4}), it is natural to treat $\psi$ as a parameter and search for a solution of the form $\xi = \xi(\psi)$. From (\ref{xi2_mot}) and (\ref{psi2_mot}) we get immediately
\begin{equation}
\label{Eq_mot}
    \frac{d\xi}{d\psi} = \epsilon_r \frac{\xi^2}{\lambda} \sqrt{\varepsilon^2 - U_\lambda(\xi)} = \epsilon_r \sqrt{\frac{\varepsilon^2 -1}{\lambda^2} \xi^4 + \frac{2}{\lambda^2} \xi^3 - \xi^2 + 2 \xi}.
\end{equation}
Defining
\begin{equation}
\label{f_general}
    f(\xi) = a_0 \xi^4 + 4a_1 \xi^3 +6a_2\xi^2 + 4a_3\xi + a_4,
\end{equation}
where
\begin{equation}
\label{fcoeffs}
    a_0 = \frac{\varepsilon^2 -1}{\lambda^2}, \quad 4a_1 = \frac{2}{\lambda^2}, \quad 6a_2 = - 1, \quad 4a_3 = 2, \quad a_4 = 0,
\end{equation}
one can write Eq.\ (\ref{Eq_mot}) as
\begin{equation}
\label{equationf}
   \frac{d\xi}{d\psi} =\epsilon_r \sqrt{f(\xi)}.
\end{equation}
For a segment of the trajectory for which $\epsilon_r$ is constant, we get
\begin{equation}
\label{psi_integral}
\psi = \epsilon_r \int_{\xi_0}^\xi \frac{d \xi^\prime}{\sqrt{f(\xi^\prime)}},
\end{equation}
where $\xi_0$ is an arbitrarily chosen radius corresponding to the angle $\psi = 0$. We emphasize that $\sqrt{}$ is assumed to be non-negative. Weierstrass invariants of the polynomial $f$ read (see  \ref{appendix:BW_theorem})
\begin{subequations}
 \label{invariants_phys}
\begin{eqnarray}   
        g_2 & = & \frac{1}{12} - \frac{1}{\lambda^2}, \\
        g_3 & = & \frac{1}{6^3} - \frac{1}{12 \lambda^2} - \frac{\varepsilon^2 - 1}{4 \lambda^2}.
\end{eqnarray}
\end{subequations}
Therefore, thanks to the Biermann-Weierstrass theorem (see \ref{appendix:BW_theorem} for a statement of this theorem and the proof),
we can write the formula for $\xi = \xi(\psi)$ as 
\begin{equation}
\label{xi_psi}
    \xi(\psi) =  \xi_0 + \frac{- \epsilon_r \sqrt{f(\xi_0)} \wp'(\psi) + \frac{1}{2} f'(\xi_0 ) \left[ \wp(\psi) - \frac{1}{24}f''(\xi_0 )\right] + \frac{1}{24} f(\xi_0 ) f'''(\xi_0 )  }{2 \left[ \wp(\psi) - \frac{1}{24} f''(\xi_0 ) \right]^2 - \frac{1}{48} f(\xi_0 ) f^{(4)}(\xi_0 ) }.
\end{equation}
Here $\wp$ is understood to be defined by the invariants $g_2$, and $g_3$ given by Eq.\ (\ref{invariants_phys}), i.e., $\wp(z) = \wp(z;g_2,g_3)$, and $f$ is defined in Eqs.\ (\ref{f_general}) and (\ref{fcoeffs}).

We emphasize that formula (\ref{xi_psi}) works in a much more general setting than described above. It turns out to be valid also for trajectories along which the sign $\epsilon_r$ changes. This can be checked numerically, but there is also a way to demonstrate this fact analytically. The argument can be sketched as follows.

Denote the functions defined by Eq.\ (\ref{xi_psi}) and corresponding to two different signs $\epsilon_r$ as
\begin{equation}
\label{xi_psi_minus}
    \xi_{-}(\psi; \xi_0) =  \xi_0 + \frac{+ \sqrt{f(\xi_0)} \wp'(\psi) + \frac{1}{2} f'(\xi_0 ) \left[ \wp(\psi) - \frac{1}{24}f''(\xi_0 )\right] + \frac{1}{24} f(\xi_0 ) f'''(\xi_0 )  }{2 \left[ \wp(\psi) - \frac{1}{24} f''(\xi_0 ) \right]^2 - \frac{1}{48} f(\xi_0 ) f^{(4)}(\xi_0 ) }
\end{equation}
and
\begin{equation}
\label{xi_psi_plus}
    \xi_{+}(\psi; \xi_0) =  \xi_0 + \frac{- \sqrt{f(\xi_0)} \wp'(\psi) + \frac{1}{2} f'(\xi_0 ) \left[ \wp(\psi) - \frac{1}{24}f''(\xi_0 )\right] + \frac{1}{24} f(\xi_0 ) f'''(\xi_0 )  }{2 \left[ \wp(\psi) - \frac{1}{24} f''(\xi_0 ) \right]^2 - \frac{1}{48} f(\xi_0 ) f^{(4)}(\xi_0 ) }.
\end{equation}
It follows from Eq.\ (\ref{psi_integral}) that $\xi_{-}(\psi;\xi_0) = \xi_{+}(-\psi;\xi_0)$.

Consider a particle moving initially inwards (i.e., with $\epsilon_r = -1$) from a starting position $\xi_0$ to the turning point $\xi_1$, for which $f(\xi_1) = 0$, and then moving outwards (with $\epsilon_r = +1$) up to a point with the radius $\xi$. The angle $\psi$ corresponding to this motion can be expressed as $\psi = \psi_1 + \psi_2$, where
\begin{equation}
\label{psi1psi2defs}
\psi_1 = - \int_{\xi_0}^{\xi_1} \frac{d \xi^\prime}{\sqrt{f(\xi^\prime)}} = \int_{\xi_1}^{\xi_0} \frac{d \xi^\prime}{\sqrt{f(\xi^\prime)}}, \quad \psi_2 = \int_{\xi_1}^{\xi} \frac{d \xi^\prime}{\sqrt{f(\xi^\prime)}}.
\end{equation}
For both angles $\psi_1$ and $\psi_2$ we have, according to the Biermann-Weierstrass theorem [Eqs.\ (\ref{BW_wp})]:
\begin{subequations}
\label{wp_psi1_psi2}
\begin{eqnarray}
\wp(\psi_1) & = &  \frac{f^\prime(\xi_1)}{4(\xi_0 - \xi_1)} + \frac{f^{\prime\prime}(\xi_1)}{24}, \\
\wp(\psi_2) & = & \frac{f^\prime(\xi_1)}{4(\xi - \xi_1)} + \frac{f^{\prime\prime}(\xi_1)}{24}, \\
\wp^\prime(\psi_1) & = & - \frac{f^\prime(\xi_1) \sqrt{f(\xi_0)}}{4(\xi_0 - \xi_1)^2}, \\
\wp^\prime(\psi_2) & = & - \frac{f^\prime(\xi_1) \sqrt{f(\xi)}}{4(\xi - \xi_1)^2}.
\end{eqnarray}
\end{subequations}
The simplicity of the above formulas is, of course, due to the fact that $f(\xi_1) = 0$. Using expression (\ref{xi_psi_minus}) we get $\xi_1 = \xi_{-}(\psi_1;\xi_0)$. The fact that the formula (\ref{xi_psi_minus}) describes the continuation of the trajectory in the segment from $\xi_1$ to $\xi$ means that
\begin{equation}
\label{addition_xi}
\xi = \xi_{-}(\psi_1 + \psi_2; \xi_0) = \xi_{+}(\psi_2; \xi_1) = \xi_{-}(-\psi_2; \xi_1).
 \end{equation}
While the above expression could, in principle, be checked directly, it is much easier to check the corresponding relations involving Weierstrass $\wp$ functions. According to the addition theorem for the Weierstrass elliptic function $\wp$, we have
\begin{equation}
    \wp(\psi_1 - \psi_2) = \frac{1}{4} \left[ \frac{\wp^\prime(\psi_1) + \wp^\prime(\psi_2)}{\wp(\psi_1) - \wp(\psi_2)} \right]^2 - \wp(\psi_1) - \wp(\psi_2).
\end{equation}    
Inserting in the above formula the expressions for $\wp(\psi_1)$, $\wp(\psi_2)$, $\wp^\prime(\psi_1)$, and $\wp^\prime(\psi_2)$ given by Eqs.\ (\ref{wp_psi1_psi2}), we get, after some algebra,
\begin{equation}
\label{addition_result}
  \wp(\psi_1 - \psi_2) = \frac{\sqrt{f(\xi) f(\xi_0)} + f(\xi_0)}{2(\xi-\xi_0)^2} + \frac{f^\prime(\xi_0)}{4(\xi-\xi_0)} + \frac{f^{\prime \prime}(\xi_0)}{24},
\end{equation}
as predicted by the Biermann-Weierstrass formula (\ref{BW_wp}). Deriving Eq.\ (\ref{addition_result}), we have to remember that $f$ is a fourth order polynomial given by Eq.\ (\ref{f_general}), and $f(\xi_1) = 0$. The reason for considering the difference $\psi_1 - \psi_2$, instead of the sum $\psi_1 + \psi_2$, can be understood in the light of Eq.\ (\ref{addition_xi}) and the fact that $\xi_{+}(\psi_2;\xi_1) = \xi_{-}(-\psi_2;\xi_1)$.

In summary, Eq.\ (\ref{xi_psi}) can be used to describe any orbit with $\psi = 0$ for $\xi = \xi_0$. The sign $\epsilon_r$ in Eq.\ (\ref{xi_psi}) can be understood as referring to the direction of motion at $\psi = 0$, and it need not be changed as the trajectory passes through a turning point. This stays in a clear contrast to the approaches based on Jacobi and Legendre elliptic functions, where one has to deal with different types of orbits separately. Apart from this universality, the main practical advantage of formula (\ref{xi_psi}) is the fact that it does not require finding zeros of the polynomial $f$. Of course, there are applications in which the knowledge about zeros of the polynomial $f$ is required---we require such knowledge indirectly in Secs.\ \ref{sec:rangepsi}, \ref{sec:propertime}, and \ref{sec:coordinatetime}, dealing with the allowed range of $\psi$, the proper and coordinate time $s$ and $\tau$, respectively. Note that
\begin{equation}
f(\xi) = \frac{\xi^4}{\lambda^2} \left[ \varepsilon^2 - U_\lambda(\xi) \right],
\end{equation}
and consequently zeros of the polynomial $f$ are related to zeros of the expression $\varepsilon^2 - U_\lambda(\xi)$, corresponding to turning points and discussed in \ref{sec:class_of_traj}. The Biermann-Weierstrass expression is based on a transformation of the integral appearing on the right-hand side of Eq.\ (\ref{psi_integral}) to the Weierstrass form, i.e.,
\begin{equation}
    \int_{\xi_0}^\xi \frac{d \xi^\prime}{\sqrt{f(\xi^\prime)}} = \pm \int^{\infty}_{w(x)} \frac{ dw^\prime}{\sqrt{4 {w^\prime}^3 - g_2 w^\prime - g_3}}
\end{equation}
(see \ref{appendix:BW_theorem}). Zeros of the polynomial $W = 4w^3 - g_2 w - g_3$ depend on the sign of the discriminant $\Delta = g_2^3 - 27 g_3^2$. The case with $\Delta = 0$ corresponds to $\lambda = \lambda_\mathrm{c}(\varepsilon)$, defined by Eq.\ (\ref{lambda-crit-unbound}).

%Note that although the parameter $\psi$ has a clear interpretation of an angle (the true anomaly), the function $\xi(\psi)$ given by Eq.\ (\ref{xi_psi}) does not have to be periodic with the period equal to $2\pi$. For instance, an orbit of a particle incoming from infinity, revolving more then once around the black hole and then scattered to infinity, corresponds to a period larger than $2 \pi$.

Figures \ref{figIIa}--\ref{figIIIb} show various kinds of orbits obtained with the help of Eq.\ (\ref{xi_psi}). Figure \ref{figIIa} depicts examples of bound inner orbits. Figure \ref{figIIb} shows a sample outer bound orbit. Unbound absorbed orbits are shown in Fig.\ \ref{figIIIa}. Finally, a family of unbound scattered orbits is plotted in Fig.\ \ref{figIIIb}. In all figures, the left panel depicts the radius $\xi$ versus the angle $\psi$. Right panels show the orbits in the orbital plane with Cartesian coordinates $x$, $y$. For comparison, in all cases we draw the same orbits obtained by integrating geodesic equations numerically. These numerical results are depicted with dotted or dashed lines.

\subsection{Null geodesics}

The reasoning for null geodesics is analogous. The equation defining the trajectory reads
\begin{equation}
    \frac{d \xi}{d \psi} = \epsilon_r \sqrt{\frac{\varepsilon^2}{\lambda^2} \xi^4 - \xi^2 + 2 \xi}.
\end{equation}
Adhering to the same notation as for timelike orbits, we set
\begin{equation}
\label{fnull}
    f(\xi) = a_0 \xi^4 + 4a_1 \xi^3 +6a_2\xi^2 + 4a_3\xi + a_4 = \frac{\varepsilon^2}{\lambda^2} \xi^4 - \xi^2 + 2 \xi,
\end{equation}
i.e.,
\begin{equation}
    a_0 = \frac{\varepsilon^2}{\lambda^2}, \quad a_2 = - \frac{1}{6}, \quad a_3 = \frac{1}{2},
\end{equation}
and $a_1 = a_4 = 0$. The Weierstrass invariants can be written as
\begin{subequations}
\label{weierstrassinvnull}
\begin{eqnarray}
    g_2 & = & \frac{1}{12}, \\
    g_3 & = & \frac{1}{216} - \frac{\varepsilon^2}{4 \lambda^2}.
\end{eqnarray}
\end{subequations}
With these modifications, remaining equations of the previous subsection hold for null geodesics as well. In particular, Eq.\ (\ref{xi_psi}), with $f(\xi)$ and the Weirestrass invariants given by Eqs.\ (\ref{fnull}) and (\ref{weierstrassinvnull}), is valid also for null geodesics.

\begin{figure}[t]
\centering
\includegraphics[width=0.45\linewidth]{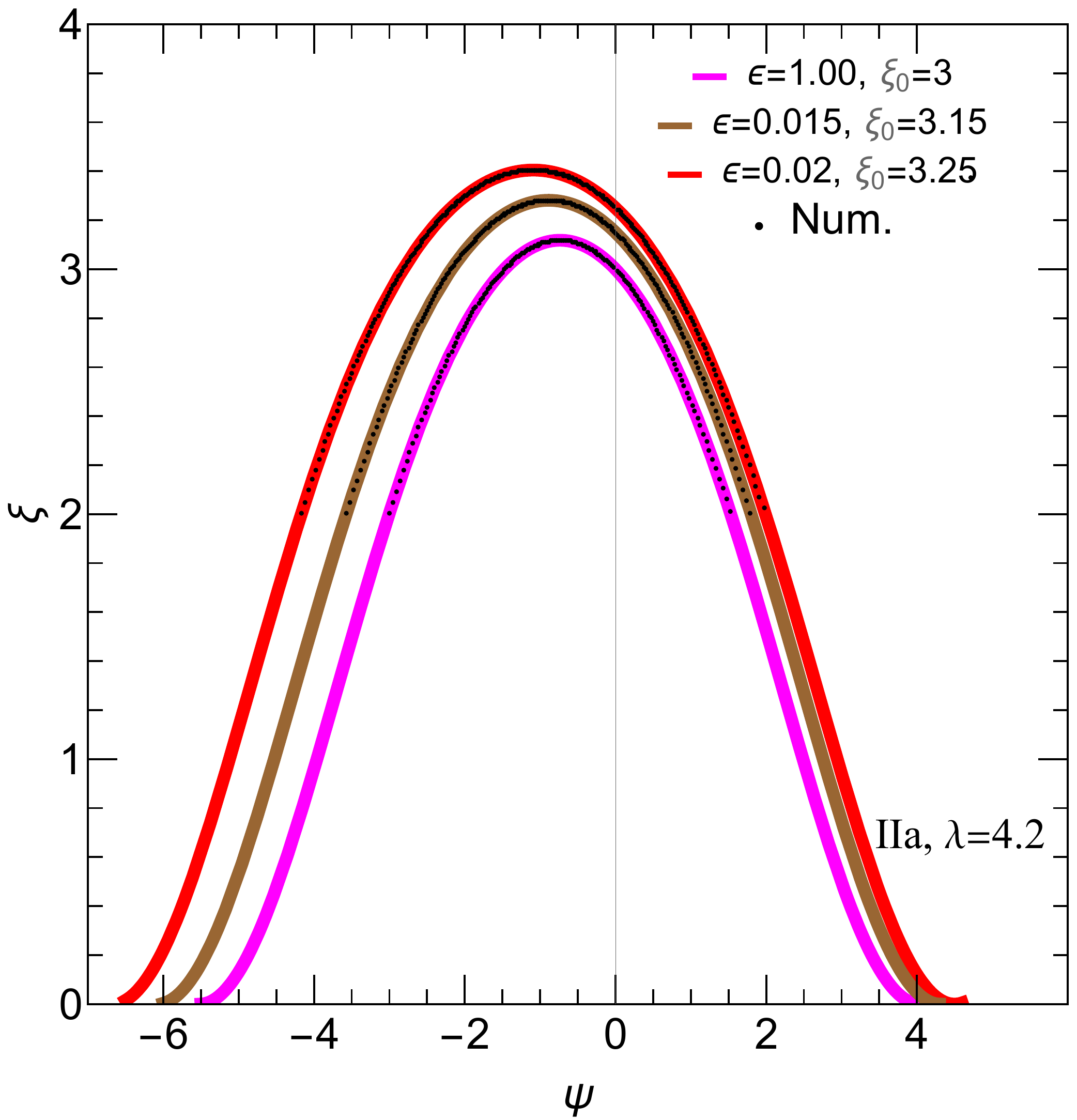}
\includegraphics[width=0.45\linewidth]{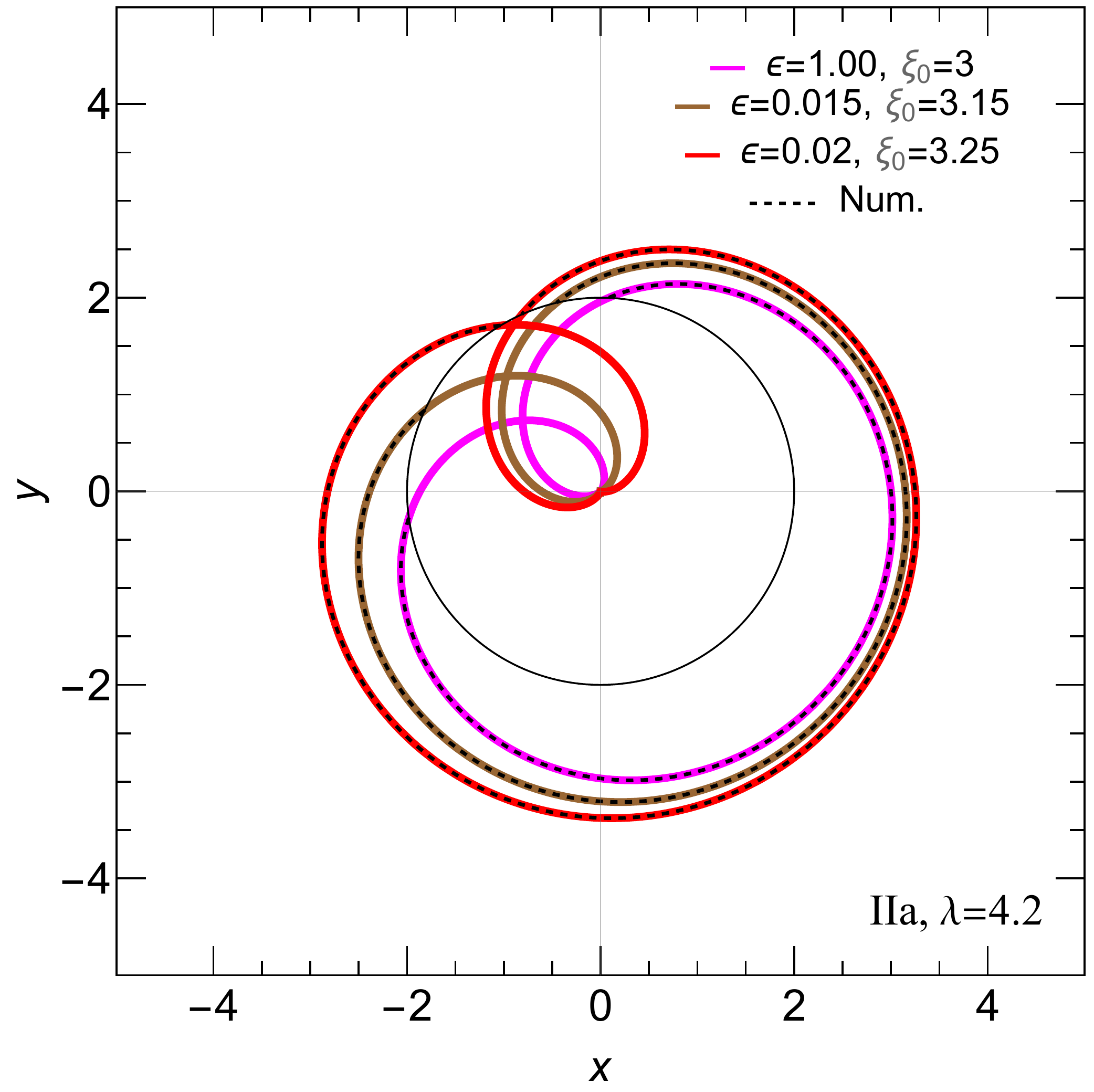}
\caption{\label{figIIa} Sample of timelike inner bound orbits (type IIa) for $\lambda = 4.2$. Solid color lines correspond to solutions obtained with Eq.\ (\ref{xi_psi}). Dotted lines depict corresponding numerical solutions.}
\end{figure}

\begin{figure}[t]
\centering
\includegraphics[width=0.45\linewidth]{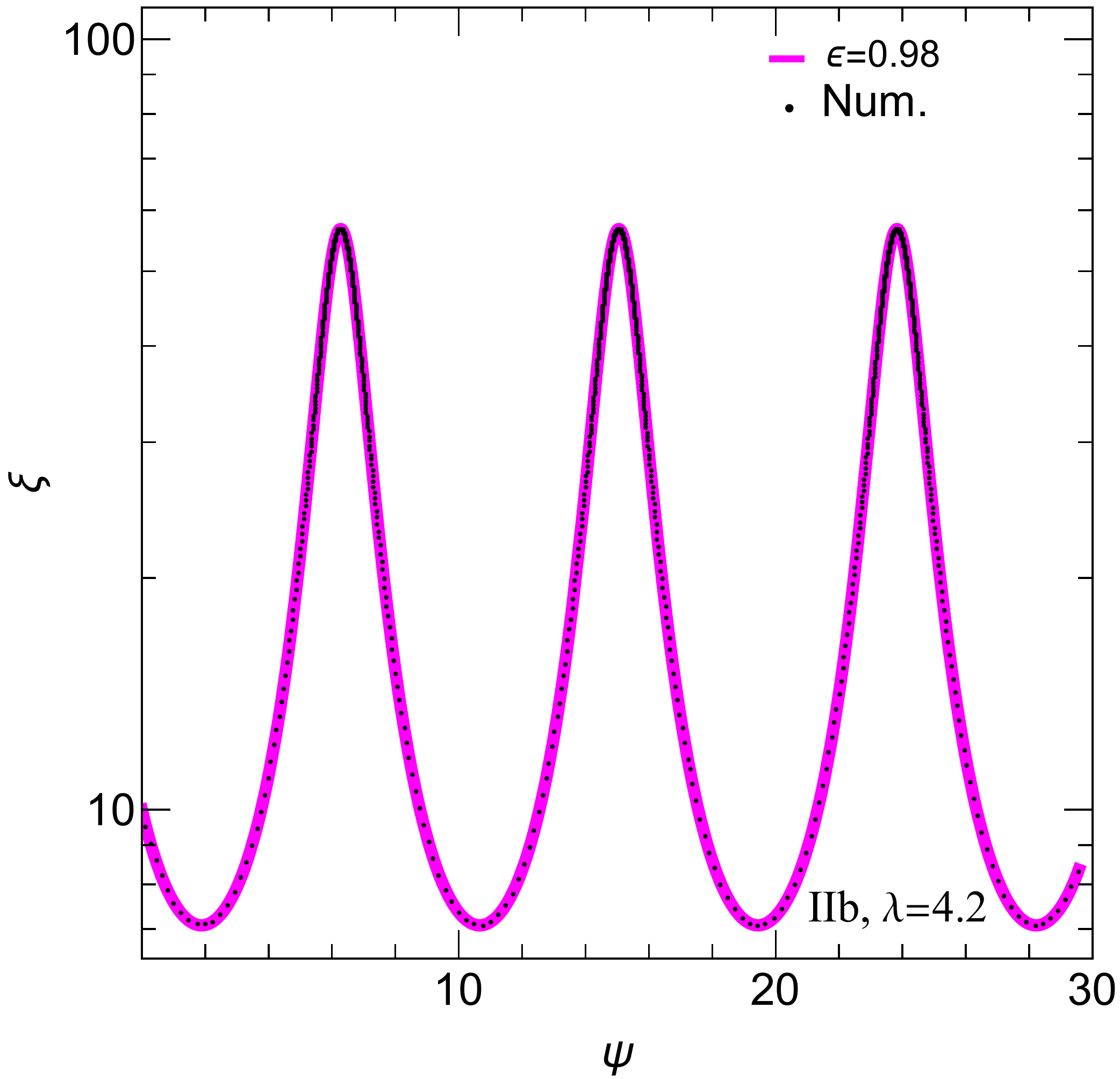}
\includegraphics[width=0.45\linewidth]{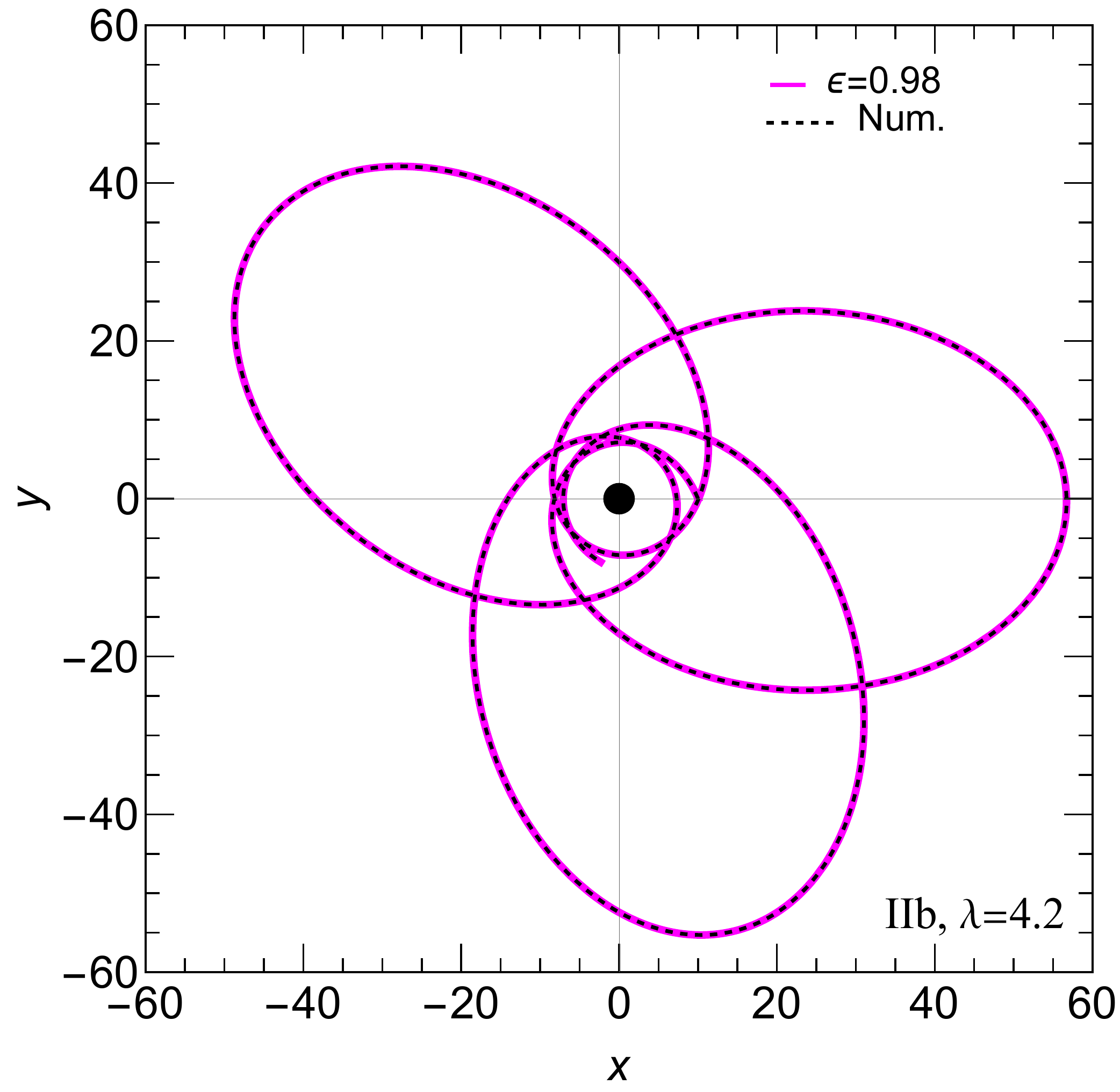}
\caption{\label{figIIb} Sample of timelike outer bound orbits (type IIb) for $\lambda = 4.2$. Solid color lines correspond to solutions obtained with Eq.\ (\ref{xi_psi}). Dotted lines depict corresponding numerical solutions.}
\end{figure}

\begin{figure}[t]
\centering
\includegraphics[width=0.45\linewidth]{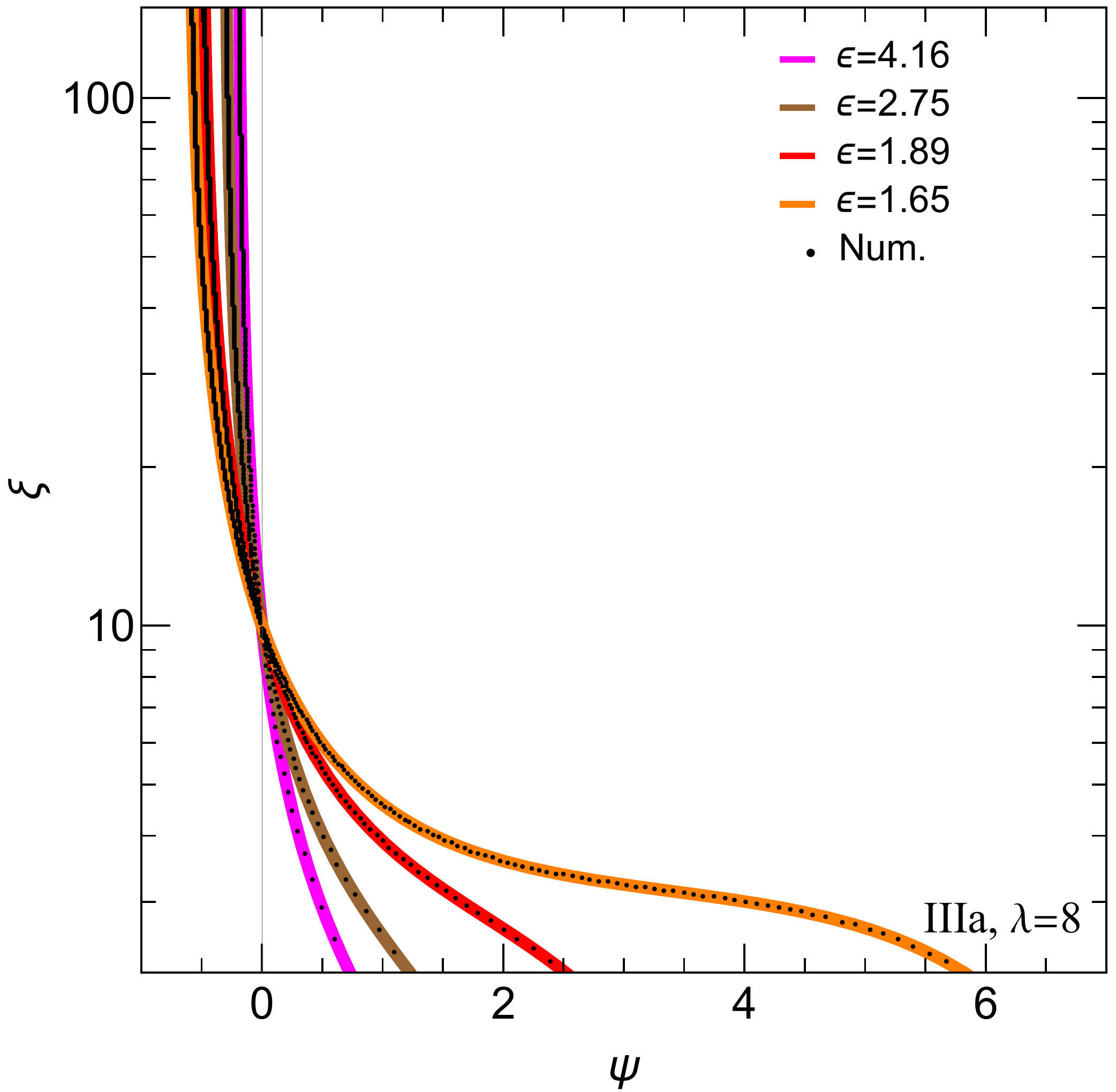}
\includegraphics[width=0.45\linewidth]{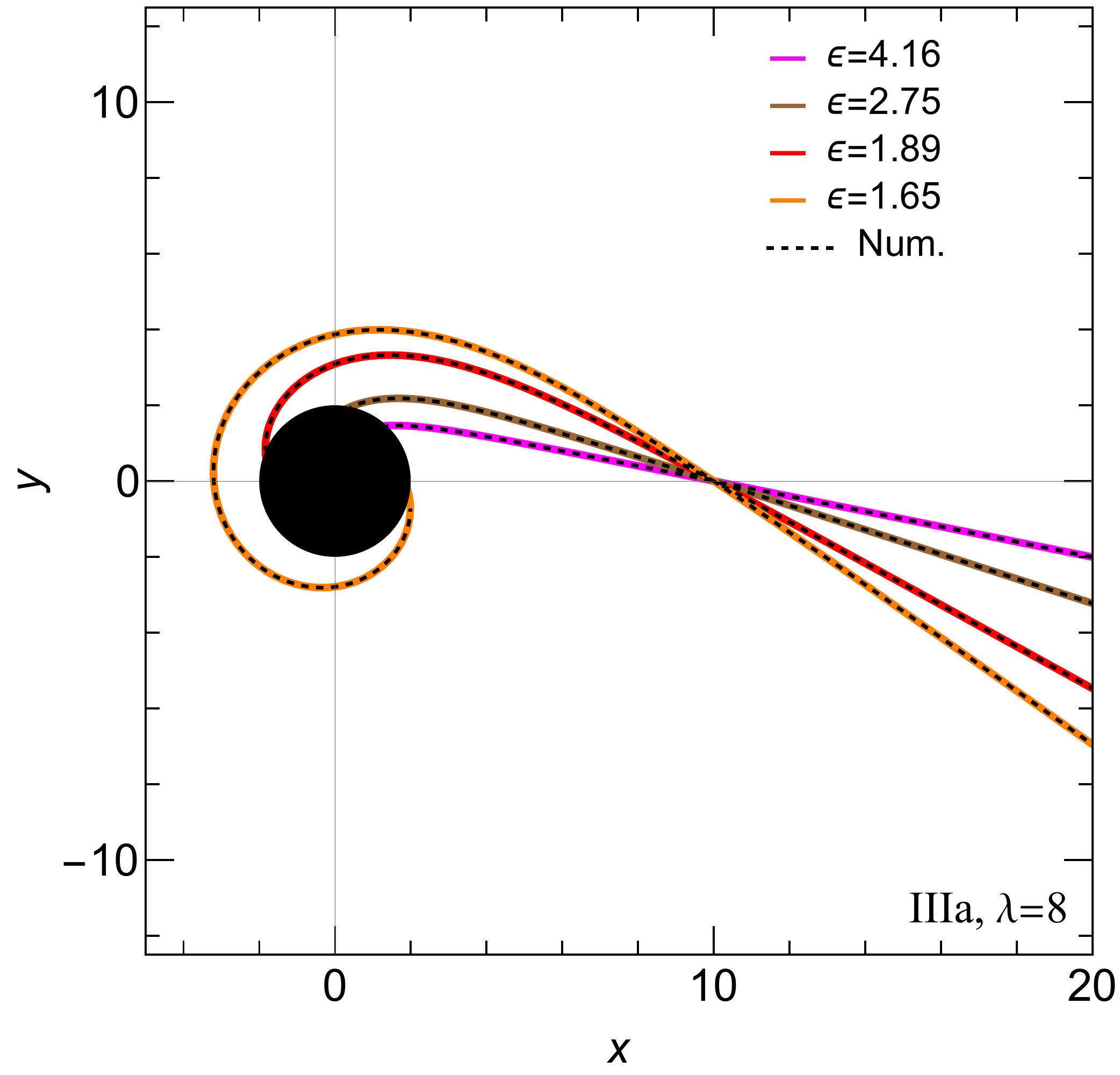}
\caption{\label{figIIIa} Sample of timelike unbound absorbed orbits (type IIIa) for $\lambda = 8$. Solid color lines correspond to solutions obtained with Eq.\ (\ref{xi_psi}). Dotted lines depict corresponding numerical solutions.}
\end{figure}

\begin{figure}[t]
\centering
\includegraphics[width=0.45\linewidth]{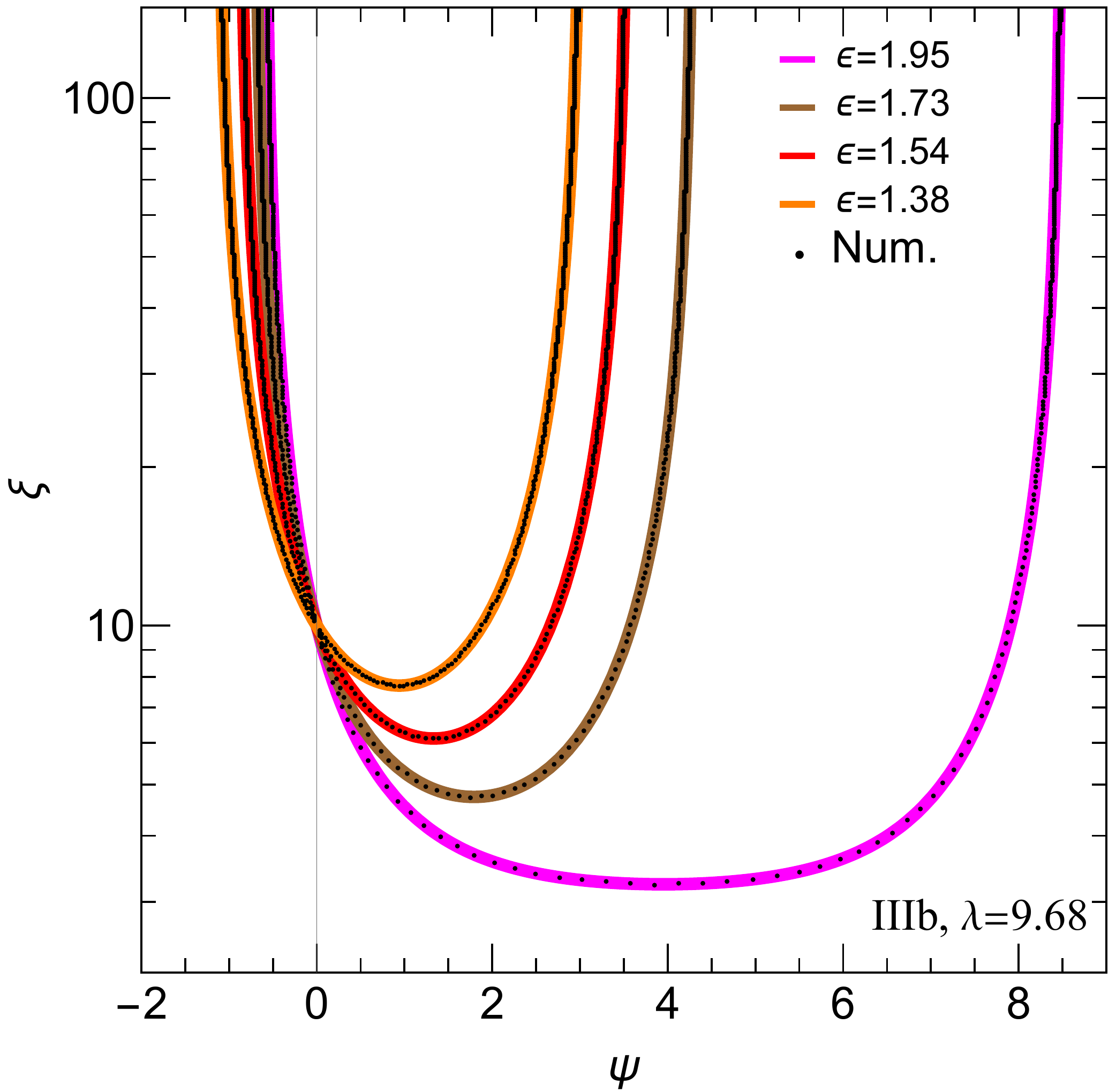}
\includegraphics[width=0.45\linewidth]{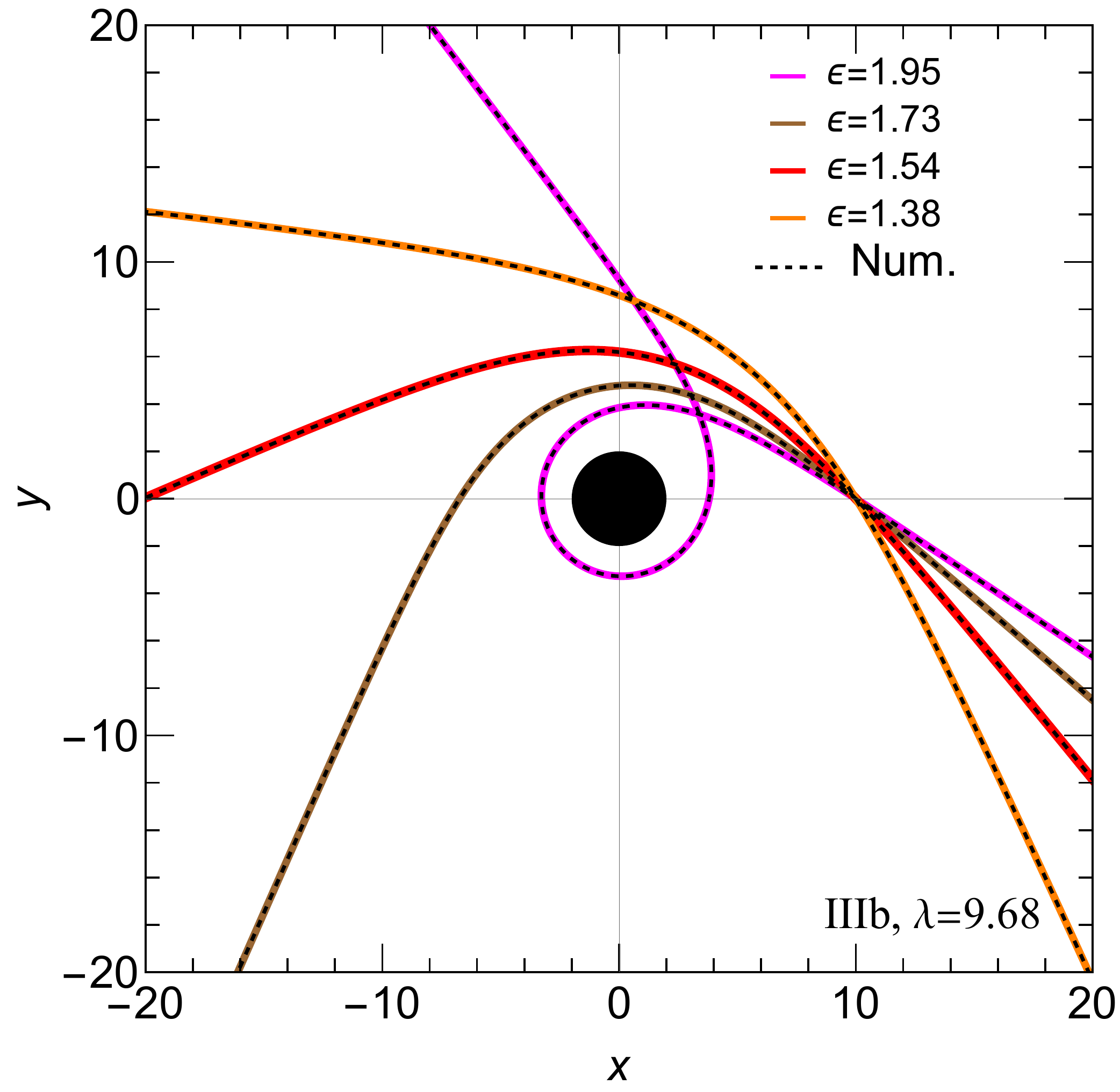}
\caption{\label{figIIIb} Sample of timelike  unobound scattered orbits (type IIIb) for $\lambda = 9.68$. Solid color lines correspond to solutions obtained with Eq.\ (\ref{xi_psi}). Dotted lines depict corresponding numerical solutions.}
\end{figure}

\begin{figure}[t]
\centering
\includegraphics[width=0.45\linewidth]{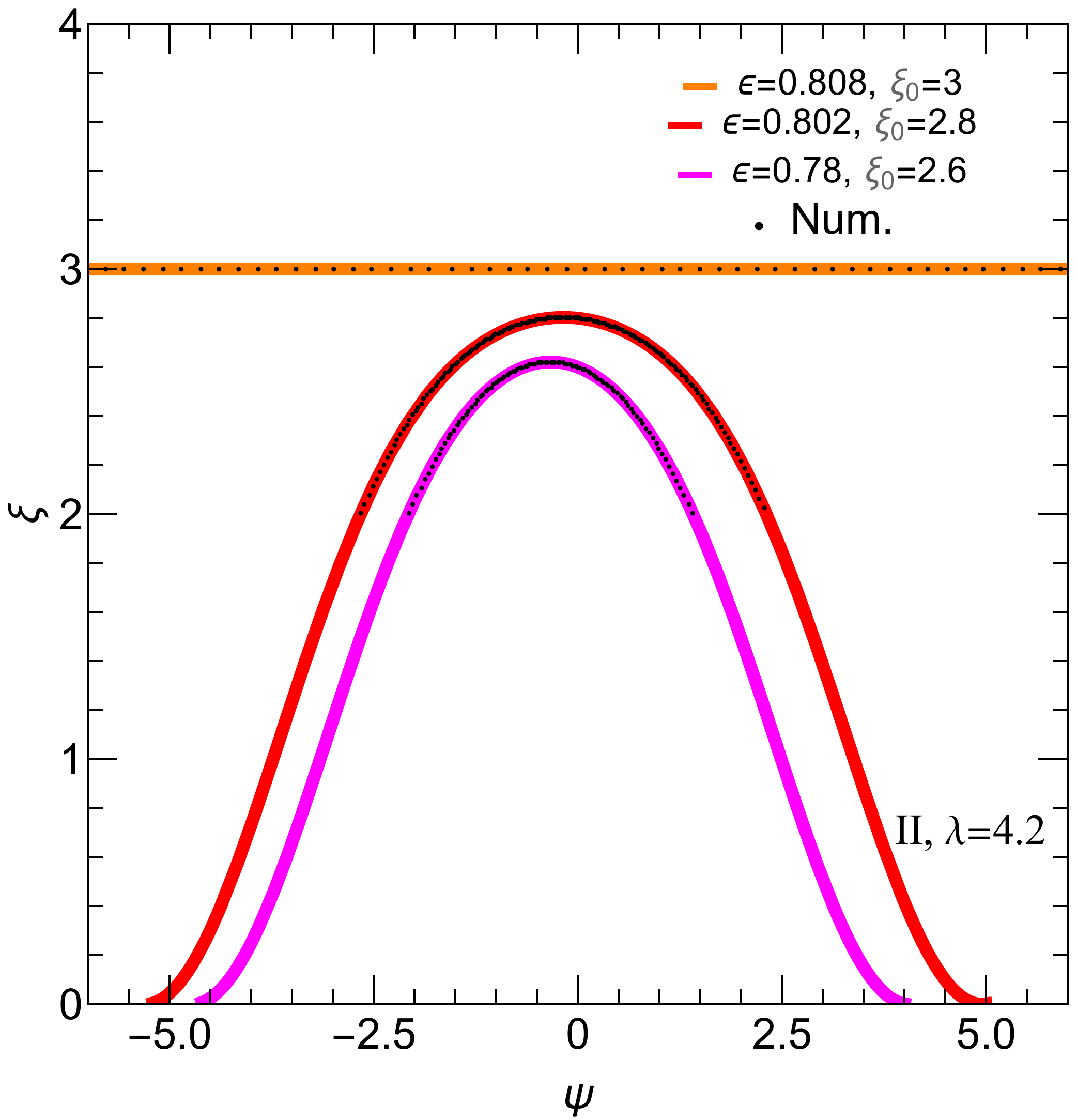}
\includegraphics[width=0.45\linewidth]{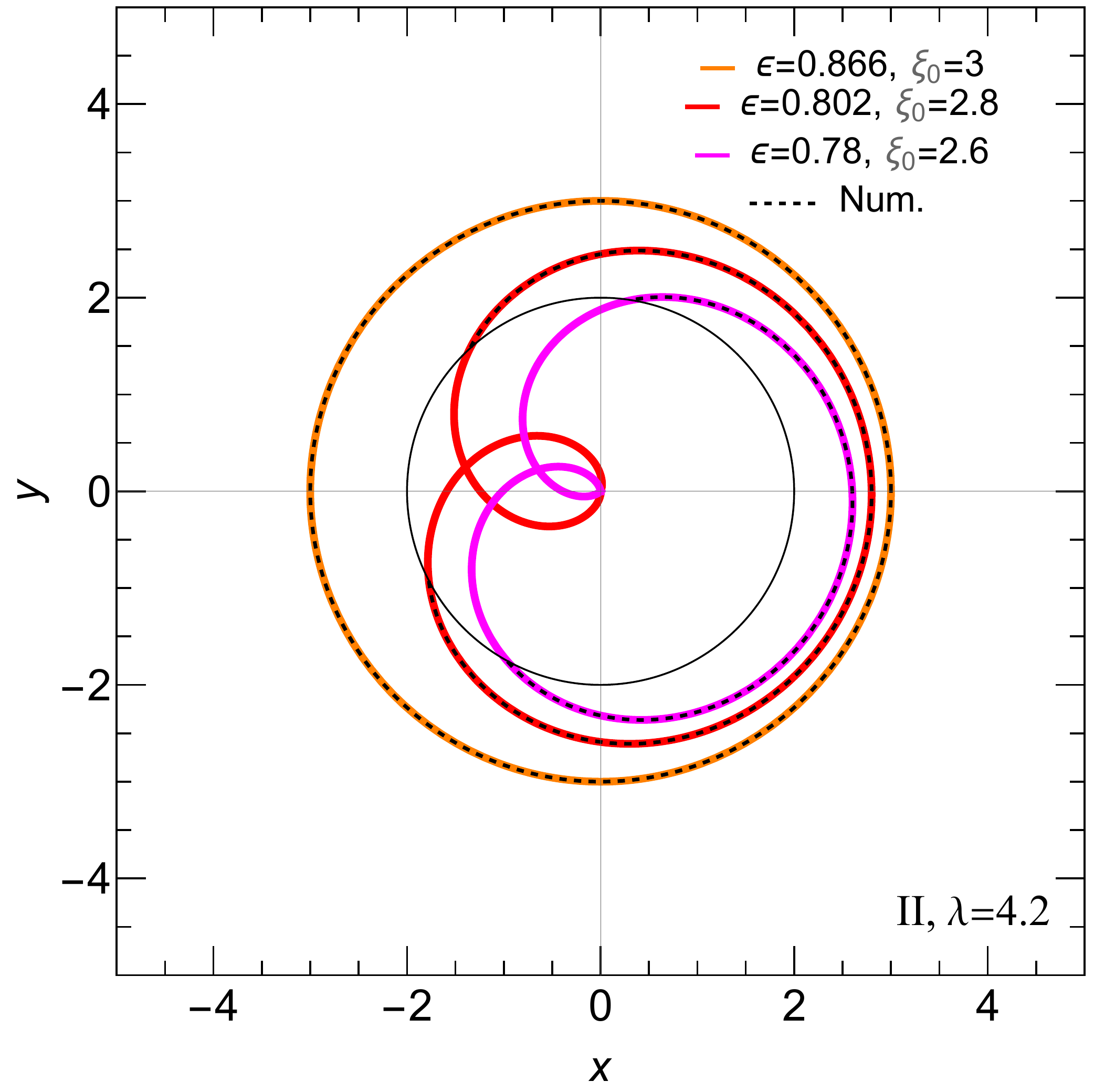}
\caption{\label{figII_zero} Sample of null bound orbits (type II) for $\lambda = 4.2$. Solid color lines correspond to solutions obtained with Eq.\ (\ref{xi_psi}). Dotted lines depict corresponding numerical solutions.}
\end{figure}

\begin{figure}[t]
\centering
\includegraphics[width=0.45\linewidth]{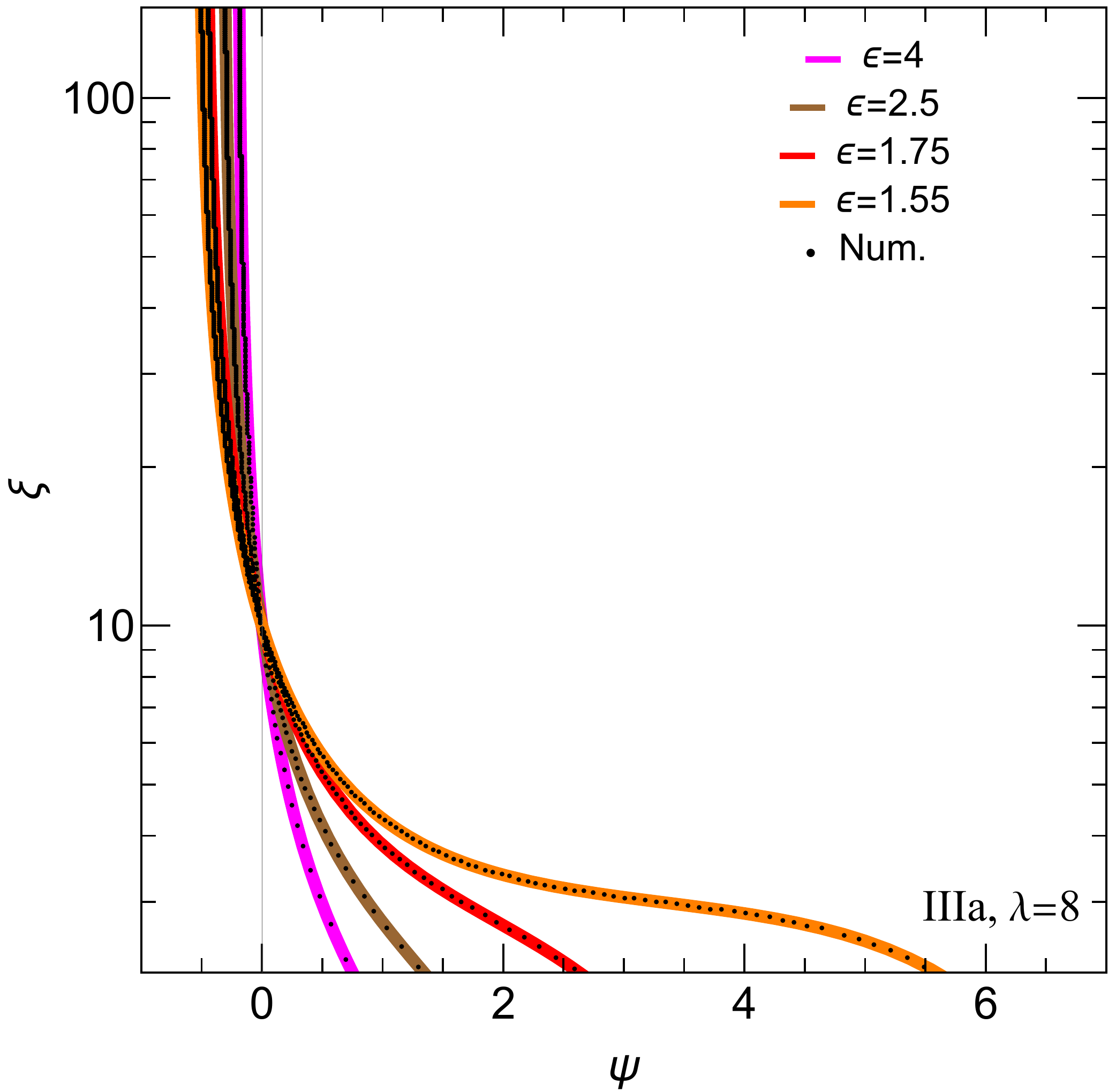}
\includegraphics[width=0.45\linewidth]{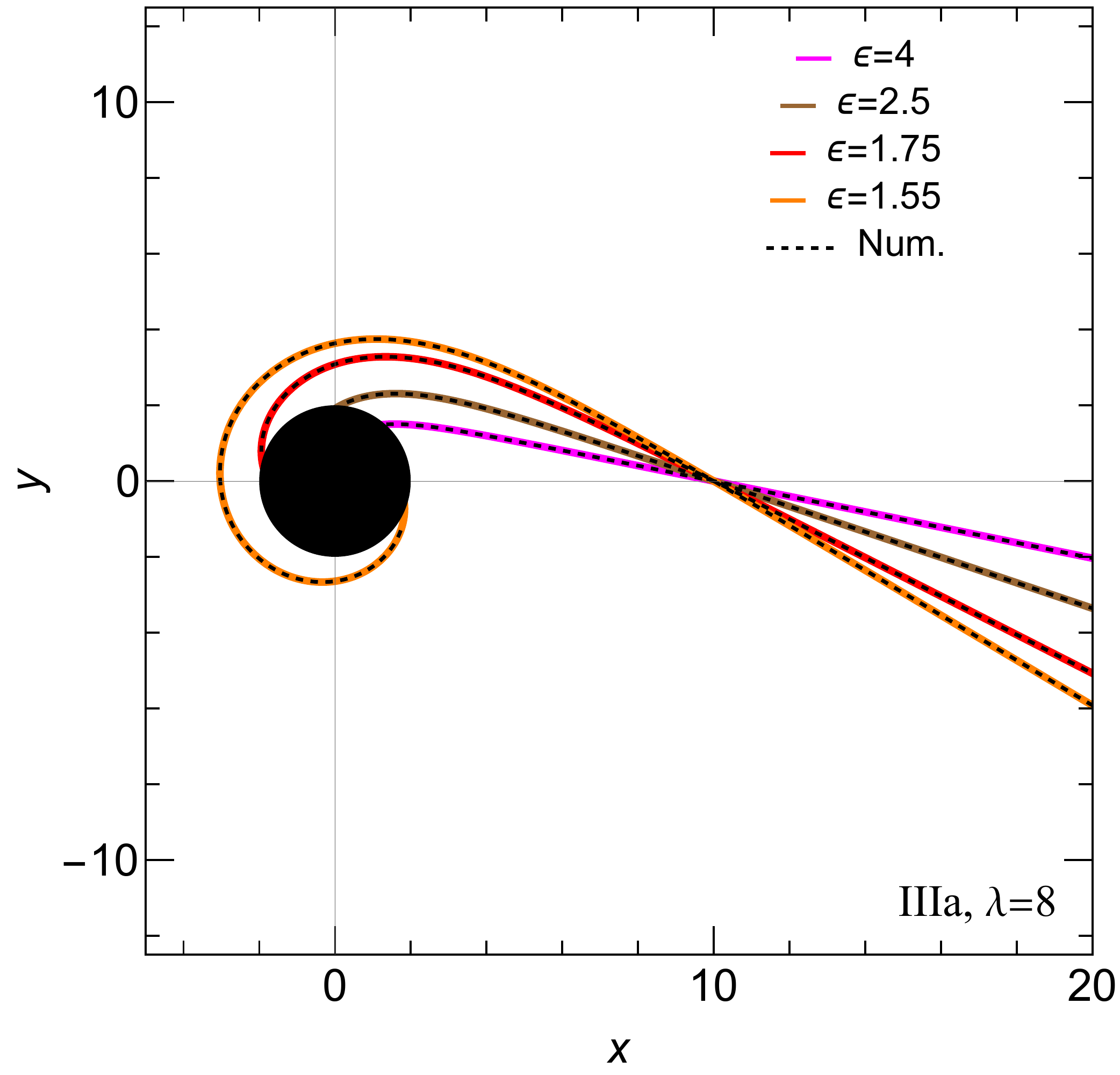}
\caption{\label{figIIIa_zero} Sample of null unbound absorbed orbits (type IIIa) for $\lambda = 8$. Solid color lines correspond to solutions obtained with Eq.\ (\ref{xi_psi}). Dotted lines depict corresponding numerical solutions.}
\end{figure}

\begin{figure}[t]
\centering
\includegraphics[width=0.45\linewidth]{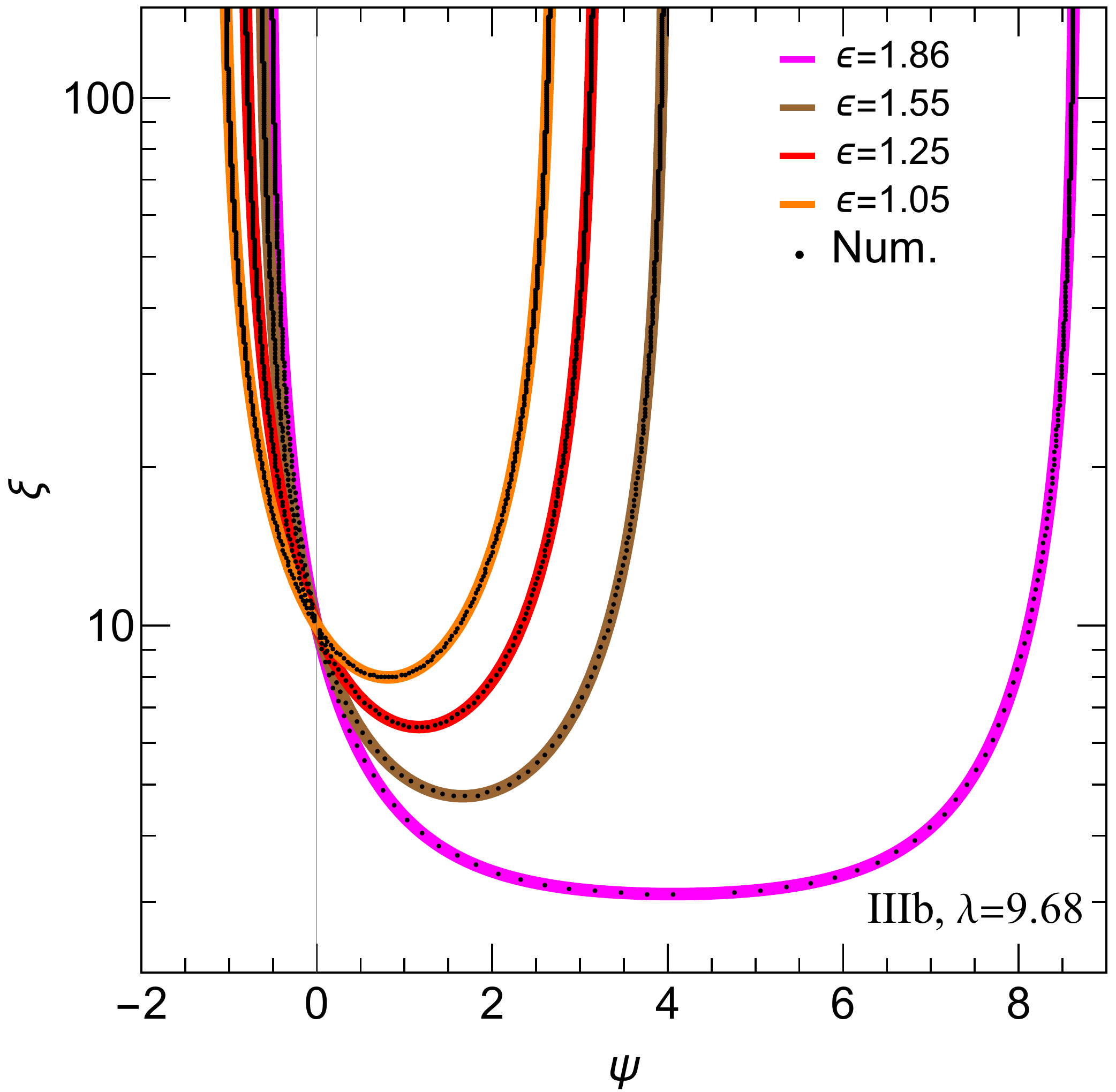}
\includegraphics[width=0.45\linewidth]{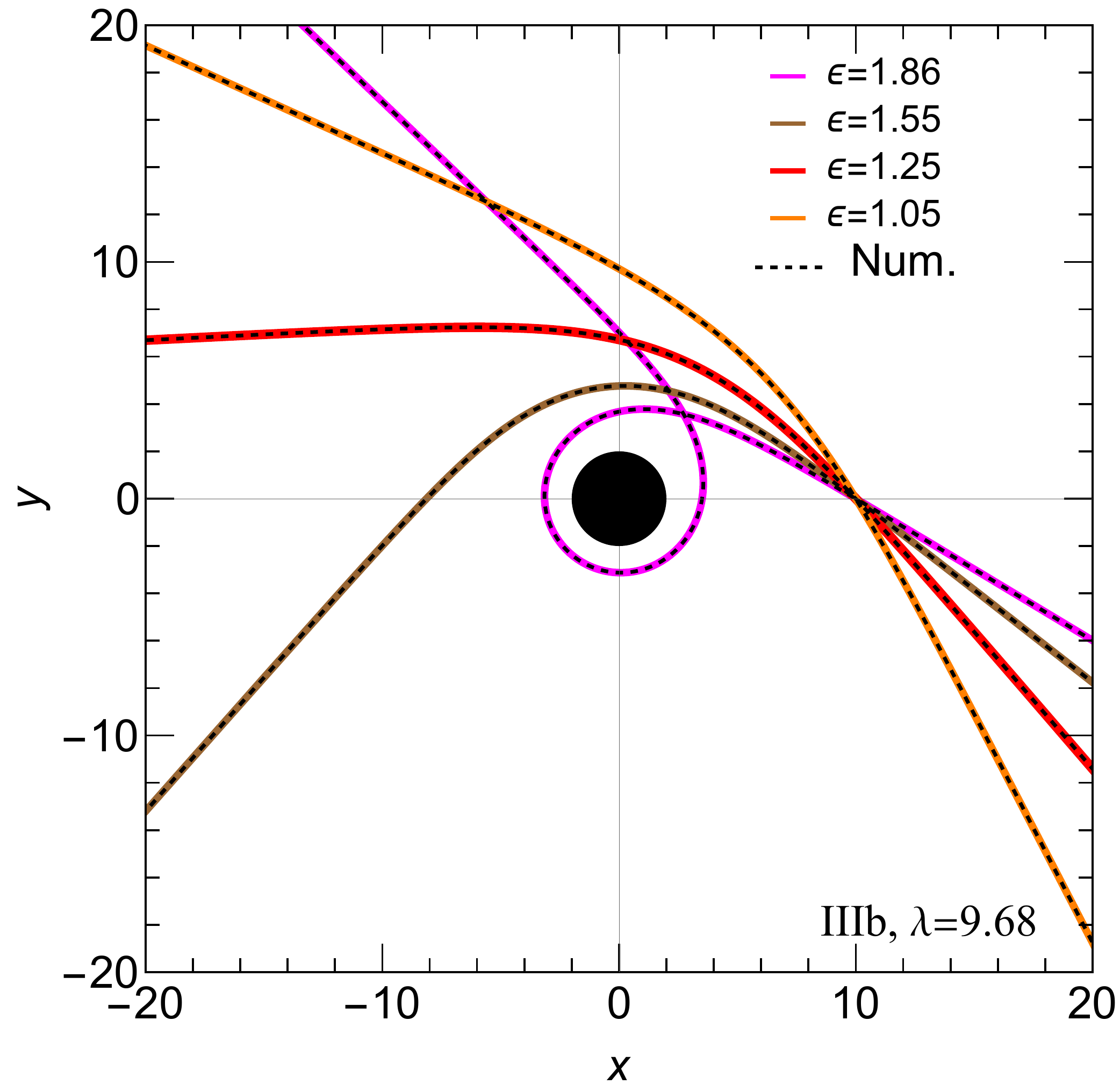}
\caption{\label{figIIIb_zero} Sample of null unbound absorbed orbits (type IIIa) for $\lambda = 9.68$. Solid color lines correspond to solutions obtained with Eq.\ (\ref{xi_psi}). Dotted lines depict corresponding numerical solutions.}
\end{figure}

Figures \ref{figII_zero}--\ref{figIIIb_zero} show a sample of null trajectories obtained with Eq.\ (\ref{xi_psi}). As for timelike geodesics, dotted lines depict solutions obtained numerically. 

\section{The range of $\psi$}
\label{sec:rangepsi}

In practical applications, one may need to control the allowed range of the parameter $\psi$ or to compute the values of $\psi$ referring to certain points at the trajectory (pericenter, apocenter).

In particular, for unbound scattered orbits the angles $\psi_{\infty \pm}$ corresponding to the asymptotics $\xi \to \infty$ could be obtained from the equation
\begin{equation}
\label{wpzeros}
    \wp(\psi_{\infty_{\pm}}) = \frac{1}{24} f''(\xi_0) \pm \sqrt{\frac{1}{96} f(\xi_0) f^{(4)}(\xi_0)},
\end{equation}
i.e., from the requirement that the denominator in Eq.\ (\ref{xi_psi}) vanishes.

For unbound trajectories of particles that fall into the black hole, the sign $\epsilon_r$ is constant along the trajectory. As a consequence, one can use Eq.\ (\ref{BW_wp}) of  \ref{appendix:BW_theorem} applied directly to the integral in Eq.\ (\ref{psi_integral}). This yields
\begin{eqnarray}
        \lim_{\xi \to \infty} \wp(\psi) & = &  \lim_{\xi \to \infty} \left[ \frac{\sqrt{f(\xi) f(\xi_0)} + f(\xi_0)}{2(\xi-\xi_0)^2} + \frac{f'(\xi_0)}{4(\xi-\xi_0)} + \frac{f''(\xi_0)}{24} \right] \nonumber \\
        & = & \frac{ f''(\xi_0)}{24} + \frac{\sqrt{a_0 f(\xi_0)}}{2} = \frac{1}{24} f''(\xi_0) + \sqrt{\frac{1}{96} f(\xi_0) f^{(4)}(\xi_0)} \nonumber \\
        & = & \wp\left(\psi_{\infty_{+}}\right),
\end{eqnarray}
meaning that $\psi_{\infty_+}$ is the relevant angle in this case.

Note that in order to get $\psi_{\infty \pm}$ direcly form Eq.\ (\ref{wpzeros}), one would have to invert (locally) the Weierstrass function $\wp$, which is troublesome in practical applications, as $\wp$ is not a one to one map.

In \ref{appendix:psi_xi}, we express the function
\begin{equation}
    X(\xi_0) = \int_{\xi_0}^\infty \frac{d \xi}{\sqrt{f(\xi)}}
\end{equation}
for an unbound scattered timelike or null trajectory in terms of the Legendre elliptic integrals. The result reads
\begin{equation}
    X(\xi) = \frac{1}{\sqrt{y_3 - y_1}} \left[ F \left( \arccos \sqrt{ \frac{y_2 + \frac{1}{12} - \frac{1}{2\xi}}{y_2 - y_1}} , k \right) - F \left( \arccos \sqrt{\frac{y_2 + \frac{1}{12}}{y_2 - y_1}}, k \right) \right],
\end{equation}
where $y_1 < y_2 < y_3$ are real zeros of the polynomial $4 y^3 - g_2 y - g_3$, and $k^2 = (y_2 - y_1)/(y_3 - y_1)$.

For a particle arriving from infinity at a scattered trajectory, the angles $\psi_{\infty \pm}$ can be obtained in a way illustrated in Fig.\ \ref{fig:angles_at_infinity}. Let $\xi_0$ be a location of an incoming particle, with $\epsilon_r = -1$, corresponding to $\psi = 0$. The angle $\psi_{\infty +} < 0$ can be expressed as $\psi_{\infty +} = - X(\xi_0)$. Denote the location of the pericenter as $\xi_\mathrm{per}$; suppose it corresponds to $\psi = \tilde \psi$. We have
\begin{equation}
    \tilde \psi + |\psi_{\infty +}| = \tilde \psi - \psi_{\infty +} = X(\xi_\mathrm{per}).
\end{equation}
Since the orbit is symmetric with respect to $\xi_\mathrm{per}$, one can express $\psi_{\infty -}$ as 
\begin{equation}
    \psi_{\infty -} = 2 \tilde \psi + |\psi_{\infty +}| = 2 X(\xi_\mathrm{per}) - |\psi_{\infty +}| = 2 X(\xi_\mathrm{per}) - X(\xi_0).
\end{equation}

\begin{figure}[ht]
\centering
\includegraphics[width=0.55\textwidth]{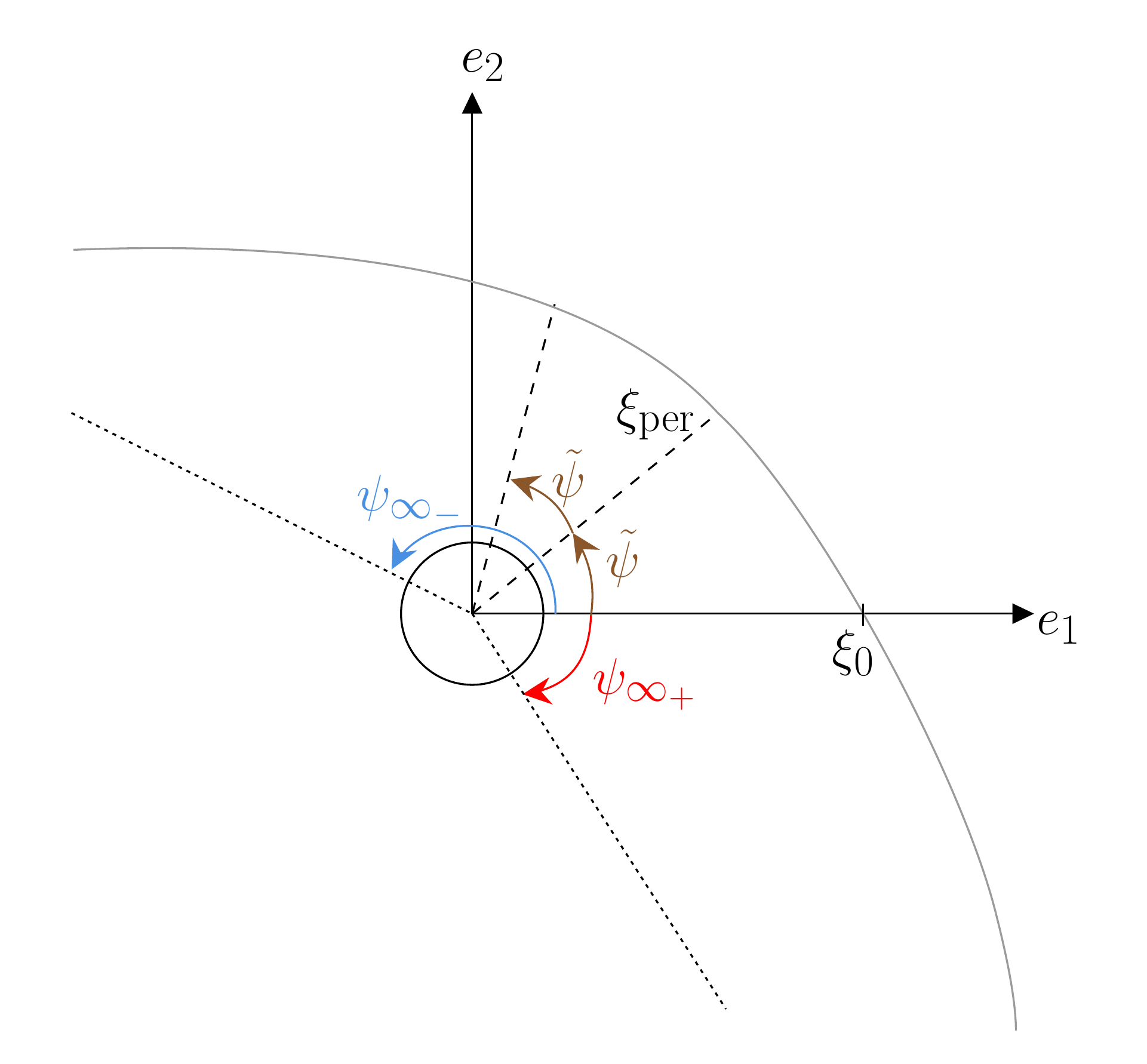}
\caption{\label{fig:angles_at_infinity}
Trajectory of a scattered particle in the motion plane. The radius $\xi_\mathrm{per}$ and the angle $\tilde \psi$ correspond to the pericenter.}
\end{figure}

For an unbound absorbed trajectory the parameter $\psi_{\infty+}$ can be expressed as before by $\psi_{\infty+} = -X(\xi_0)$ (we take $\epsilon_r = -1$), but an explicitly real expression for $X(\xi_0)$ is different, and it is given by Eqs.\ (\ref{k2bis}) and (\ref{xbis}) of  \ref{appendix:psi_xi}.

For all types of trajectories---bounded and unbounded ones---one can define the integral
\begin{equation}
    Y(\xi_0; \tilde \xi) = \int_{\xi_0}^{\tilde \xi} \frac{d\xi}{\sqrt{f(\xi)}},
\end{equation}
which we also compute in \ref{appendix:psi_xi}.

Finally note that for scattered trajectories the values $\psi_{\infty\pm}$ yield the bending (or deflection) angle in the Schwarzschild spacetime, which has been studied both for timelike and null geodesics \cite{frittelli,liu,tsupko,virbhadra}.

%=================================================================================================================================
%                                               Affine parameter 
%=================================================================================================================================

\section{Affine parameter (proper time)}
\label{sec:propertime}

In this section we compute the affine parameter $s$ associated with a given angle $\psi$. For timelike geodesics the value of $s$ is related to the proper time $\tilde \tau$ by $\tilde \tau = M s$.

Given an expression for $\xi = \xi(\psi)$, the corresponding affine parameter can be computed by integrating Eq.\ (\ref{psi2_mot}), i.e., as
\begin{equation}
\label{sdream}
    s(\psi) = \frac{1}{\lambda} \int_0^\psi \xi^2 \left( \psi^\prime \right) d\psi^\prime.
\end{equation}
Integrating the square of expression (\ref{xi_psi}) is, in principle, possible, but it is tedious, and the result seems to be too complicated to be useful in practical applications. Much simpler formulas can be obtained using Eq.\ (\ref{psi2_mot}) with the reference position taken at a zero of the polynomial $f(\xi)$.

Let $\xi_1$ denote a radius such that $f(\xi_1) = 0$ (usually a periapsis or an apoapsis); assume that it corresponds to $\psi = 0$. The radius $\xi$ corresponding to an angle $\psi$ reads
\begin{equation}
\label{xi_psi_simp}
    \xi(\psi) =  \xi_1 + \frac{ \frac{1}{4} f'(\xi_1) }{\wp(\psi) - \frac{1}{24} f''(\xi_1) }
\end{equation}
(irrespectively of the radial direction of motion, i.e., the value of $\epsilon_r$). The affine parameter elapsed during the motion from $\psi = 0$ to some $\psi = \psi_2$ can be written as 
\begin{equation}
     s_\ast(\psi_2;\xi_1) = \frac{1}{\lambda} \int_0^{\psi_2} \left\{ \xi_1^2 + \frac{\frac{1}{2}f'(\xi_1) \xi_1}{\wp(\psi) - \frac{1}{24} f''(\xi_1)} + \frac{ \frac{1}{16} \left[f'(\xi_1)\right]^2 }{\left[\wp(\psi) - \frac{1}{24} f''(\xi_1)\right]^2 } \right\} d\psi. 
\end{equation}
The above integral can be computed with the help of the following two integral formulas (\cite{byrd_handbook_1971}, p. 312 and \cite{gradshtein_table_2007}, p. 626):
\begin{equation}
    I_1 (x;y)  =  \int \frac{dx}{\wp(x) - \wp(y)} = \frac{1}{\wp'(y)} \left( 2\zeta(y) x + \ln{\frac{\sigma( x - y)}{\sigma(x + y)}} \right), \label{I_1}
\end{equation}
\begin{align}
    I_2 (x;y) & =  \int \frac{dx}{\left(\wp(x) - \wp(y)\right)^2} = - \frac{\wp''(y)}{\wp'^3(y)}\ln{\frac{\sigma\left(x - y\right)}{\sigma\left(x + y\right)}} \nonumber \\
    & -\frac{1}{\wp'^2\left( y\right)} \Bigg( \zeta\left(x +y\right) + \zeta\left(x - y\right) + \left( 2\wp\left( y\right) + \frac{2 \wp''\left( y\right)\zeta\left(y\right)}{\wp'\left(y\right)} \right) x \Bigg), \label{I_2}
\end{align}
where $\zeta(x)$ and $\sigma(x)$ denote the Weierstrass functions $\zeta(x;g_2,g_3)$ and $\sigma(x;g_2,g_3)$, respectively. We have
\begin{eqnarray}
s_\ast(\psi_2,\xi_1) & = & \frac{1}{\lambda} \left\{ \xi_1^2 \psi_2 + \frac{1}{2} f^\prime(\xi_1) \xi_1 \left[ I_1 (\psi_2;y) - I_1(0;y) \right] \right. \nonumber \\
&& \left.+ \frac{1}{16} \left[ f^\prime(\xi_1) \right]^2 \left[ I_2 (\psi_2;y) - I_2(0;y) \right] \right\},
\label{sastgotowe}
\end{eqnarray}
where $\wp(y) = \frac{1}{24} f^{\prime \prime}(\xi_1)$ or $y = \wp^{-1} \left( \frac{1}{24} f^{\prime \prime}(\xi_1) \right)$. Usually, using the inverse of the Weierstrass function is troublesome, since $\wp$ is not a one-to-one function. Fortunately, in formula (\ref{sastgotowe}), one is permitted to choose any $y$ satisfying the above condition.

We now invoke to the reasoning from the end of the previous section. Consider a motion of a particle starting from an arbitrary location $\xi_0$ and moving inwards, up to a periapsis with the radius $\xi_1$ (thus $f(\xi_1) = 0$). Next the particle moves outwards, up to a location with a radius $\xi$. Define the angles $\psi_1$ and $\psi_2$ by Eqs.\ (\ref{psi1psi2defs}). Both angles satisfy $\psi_1 \ge 0$ and $\psi_2 \ge 0$. Let $\psi = \psi_1 + \psi_2$. Because of symmetry, the proper time of the entire motion can be written as
\begin{equation}
\label{sgotowe}
s(\psi) = s_\ast(\psi_1;\xi_1) + s_\ast(\psi_2;\xi_2) = s_\ast(\psi_1;\xi_1) + s_\ast(\psi - \psi_1;\xi_1).
\end{equation}
Formula (\ref{sgotowe}) can be understood as a replacement for integral (\ref{sdream}) with $\xi(\psi)$ given by Eq.\ (\ref{xi_psi}). Note that, since $s_\ast(\psi_2;\xi_1)$ is an odd function of $\psi_2$, we get $s(\psi = 0) = 0$, as expected. It can also be checked that the same formula holds for $\xi_1$ corresponding to an apoapsis, provided that definitions (\ref{psi1psi2defs}) are changed accordingly, so that $\psi_1 \ge 0$ and $\psi_2 \ge 0$.

%=================================================================================================================================
%                                               Time parameter
%=================================================================================================================================

\section{Coordinate time}
\label{sec:coordinatetime}

The coordinate time $\tau$ can be obtained in a way similar to the calculation of the affine parameter $s$. Consider a trajectory originating at $\xi = \xi_0$, $\psi = 0$, $\tau = 0$. The coordinate time corresponding to the lapse of the parameter $\psi$ can be computed by integrating Eqs.\ (\ref{eqsofmotion4}). Combining Eqs.\ (\ref{psi2_mot}) and (\ref{tau2_mot}), one gets
\begin{equation}
\tau(\psi) = T_1(\psi) + T_2(\psi),
\end{equation}
where
\begin{equation}
\label{T_1}
    T_1 (\psi) = \frac{\varepsilon}{\lambda} \int^{\psi}_{0} \frac{\xi^2 (\psi^\prime)}{N(\xi(\psi^\prime))} d\psi^\prime = \frac{\varepsilon}{\lambda} \int^{\psi}_{0} \frac{\xi^2 (\psi^\prime)}{1 - \frac{2}{\xi(\psi^\prime)}} d\psi'
\end{equation}
and 
\begin{eqnarray}
T_2 (\psi) & = & \frac{1}{\lambda} \int^{\psi}_{0} \frac{\epsilon_r \, \xi^2(\psi^\prime)\left[1-N(\xi(\psi^\prime))\eta(\xi(\psi^\prime))\right]\sqrt{\varepsilon^2 - U_{\lambda}(\xi(\psi^\prime))}}{N(\xi(\psi^\prime))} d \psi^\prime \nonumber \\
& = &  \int^{\xi(\psi)}_{\xi_0} \left[ \frac{1}{1-\frac{2}{\xi^\prime}} - \eta(\xi^\prime) \right] d\xi^\prime.
\label{T_2}
\end{eqnarray}
The second equality in Eq.\ (\ref{T_2}) follows from Eq.\ (\ref{xi2_mot}).

The integral $T_2$ is clearly gauge-dependent. In the standard Schwarzschild coordinates $\eta(\xi) = 1/N(\xi)$, and $T_2 \equiv 0$. Of course,
\begin{equation}
    T_2(\psi) = \int^{\xi(\psi)}_{\xi_0} \left[  \frac{1}{1-\frac{2}{\xi^\prime}} - \eta(\xi^\prime)  \right] d\xi^\prime = \xi(\psi) - \xi_0 + 2\ln{\frac{\xi(\psi) - 2}{\xi_0 -2}} - \int^{\xi(\psi)}_{\xi_0} \eta(\xi^\prime) d\xi^\prime.
\end{equation}

With the help of the identity
\begin{equation}
    \frac{\xi^2}{1-\frac{2}{\xi}} = \xi^2 + 2\xi + 4 + \frac{8}{\xi -2},
\end{equation}
integral (\ref{T_1}) can be written as
\begin{equation}
    T_1(\psi) = \frac{\varepsilon}{\lambda} \left[ \int^{\psi}_{0} \xi^2(\psi^\prime) d\psi^\prime + 2\int^{\psi}_{0} \xi(\psi^\prime) d\psi^\prime + 4 \psi + 8\int^{\psi}_{0} \frac{1}{\xi(\psi^\prime)-2} d\psi^\prime \right].
\end{equation}

In analogy to the discussion of the previous section, we will start the computation of $T_1$ considering at first the special case of a trajectory originating at a turning point $\xi = \xi_1$, $\psi = 0$, such that $f(\xi_1) = 0$. In this case the radius $\xi = \xi(\psi)$ is given by Eq.\ (\ref{xi_psi_simp}). The lapse of the function $T_1$ during the motion from $\psi = 0$ to $\psi = \psi_2$ can be expressed as
\begin{eqnarray}
\label{T_1_simp}
    T_{1\ast}(\psi_2) &= & \varepsilon s_\ast(\psi_2;\xi_1) + \frac{\varepsilon}{\lambda} \left\{ 2\xi_1 \psi_2 + \frac{1}{2}f'(\xi_1) \left[I_1(\psi_2;y) - I_1(0;y)\right] + 4 \psi_2 \right\} \nonumber \\
    && + \frac{8 \varepsilon}{\lambda}\int^{\psi_2}_{0} \frac{1}{\xi_1 -2 + \frac{(1/4)f'(\xi_1)}{\wp(\psi') - (1/24)f''(\xi_1)}} d\psi' ,
\end{eqnarray}
where $\wp(y) = \frac{1}{24} f^{\prime \prime}(\xi_1)$ and the function $I_1$ is defined by Eq.\ (\ref{I_1}). The last integral can be written in the form
%\begin{fleqn}
\begin{align}
    \int^{\psi_2}_{0} \frac{1}{\xi_1 -2 + \frac{(1/4) f^\prime(\xi_1)}{\wp(\psi') - (1/24) f^{\prime \prime}(\xi_1)}} d\psi^\prime
    &=  \frac{1}{\xi_1 - 2}\int^{\psi_2}_{0} 1 - \frac{\frac{f^\prime(\xi_1)}{4(\xi_1 - 2)}}{\wp(\psi) - \wp(z)} d\psi^\prime \nonumber \\
    & = \frac{1}{\xi_1 - 2} \left\{ \psi_2 -  \frac{f^\prime(\xi_1)}{4(\xi_1 - 2)} \left[I_1(\psi_2;z) -I_1(0;z) \right]\right\},
\end{align}
%\end{fleqn}
where $\wp(z) = \frac{1}{24} f^{\prime \prime}(\xi_1) - \frac{f'(\xi_1)}{4(\xi_1 - 2)}$. Equation (\ref{T_1_simp}) can now be written as
\begin{align}
\label{T_1_simp_2}
      T_{1\ast}(\psi_2) & =  \varepsilon s_\ast(\psi_2;\xi_1) \nonumber \\
      &  +\frac{\varepsilon}{\lambda} \left\{ \frac{2\xi^2_1}{\xi_1 -2} \psi_2 + \frac{1}{2}f'(\xi_1)\left[I_1(\psi_2;y) - I_1(0;y)\right]  -  \frac{2f'(\xi_1)}{(\xi_1 - 2)^2} \left[I_1(\psi_2;z) - I_1(0;z)\right] \right\}.
\end{align}

The next step proceeds as in the previous section. Consider a particle on a trajectory originating at the radius $\xi_0$ and $\psi = 0$, moving inwards to the turning point $\xi = \xi_1$, $\psi = \psi_1$, and then continuing outwards, up to a location with an arbitrary radius $\xi = \xi(\psi)$. Let $\psi = \psi_1 + \psi_2$, $\psi_1 \ge 0$, $\psi_2 \ge 0$, where $\psi_1$ and $\psi_2$ are given by (\ref{psi1psi2defs}). Thanks to symmetry
\begin{equation}
\label{tgotowe}
T_1(\psi) = T_{1\ast}(\psi_1;\xi_1) + T_{1\ast}(\psi_2;\xi_2) =  T_{1\ast}(\psi_1;\xi_1) +  T_{1\ast}(\psi - \psi_1;\xi_1).
\end{equation}
Again, the same formula holds for a particle moving initially outwards, provided that the signs in the definitions of $\psi_1$ and $\psi_2$ are suitably adjusted.

\section{Summary}

We have revisited the theory of timelike and null geodesics in the Schwarzschild spacetime. The novel aspect of our work is the application of the Biermann-Weierstrass theorem to the description of Schwarzschild geodesics. A single formula (\ref{xi_psi}) describes all types of timelike or null geodesic orbits, except for purely radial ones. Working with a single formula gives an advantage in those applications, in which one is forced to deal with many different orbits at the same time. We should emphasize that, in contrast to standard numerical methods, Eq.\ (\ref{xi_psi}) yields exact solutions for arbitrary evolution times, even in the case of dynamically unstable orbits.

Our motivation comes from works on kinetic description of relativistic gases. Hence, we parametrize geodesics with conserved quantities (the energy and the angular momentum of the particle) and the particle initial location. Although such a parametrization is natural (and perhaps also the most popular), it might not be optimal in some applications, especially in the context of null geodesics, in which case specifying the locations of the emitter and the observer could by more convenient (cf. \cite{korzynski}).

The Biermann-Weierstrass method of this paper is fairly general, and it is deliberately presented as such in this paper. We choose as our example the Schwarzschild spacetime, but a generalization to a large class of spherically symmetric metrics is straightforward, the Reissner-Nordstr\"{o}m spacetime being one of natural possibilities. This fact opens up a variety of applications, including astrophysical ones, related to testing the nature of astrophysical black holes, both in the context of light propagation and the motion of massive particles (see, e.g.\ \cite{naked}).

\section*{Acknowledgments}

We would like to thank anonymous referees for useful comments and suggestions. A.\ C.\ acknowledges a support of the Faculty of Physics, Astronomy and Applied Computer Science of the Jagiellonian University, grant No.\ N17/MNS/000051. P.\ M.\ was partially supported by the Polish National Science Centre Grant No.\ 2017/26/A/ST2/00530.

%=================================================================================================================================
%                                                        APPENDIX
%=================================================================================================================================

\appendix

%=================================================================================================================================
%                                               App The Biermann-Weierstrass theorem
%=================================================================================================================================

\section{Biermann-Weierstrass theorem }
\label{appendix:BW_theorem}

The proofs of Lemma \ref{lemma1a} and Theorem \ref{Biermann-Weierstrass} given below are adapted from Refs.\ \cite{biermann_problemata_1865, greenhill_applications_1892, reynolds_exact_1989}.

For any quartic polynomial
\begin{equation}
\label{quartic}
    f(x) = a_0 x^4 + 4 a_1 x^3 +6 a_2 x^2 + 4a_3 x + a_4,
\end{equation}
we express its Weierstrass invariants (\cite{whittaker_course_1927}, p.\ 453) as
\begin{subequations}
\label{invariants}
\begin{eqnarray}
g_2 & \equiv & a_0 a_4 - 4a_1 a_3 + 3 a_2^2, \\
g_3 & \equiv &  a_0 a_2 a_4 + 2a_1 a_2 a_3 -a_2^3 -a_0 a_3^2 - a_1^2 a_4.
\end{eqnarray}
\end{subequations}
The Weierstrass elliptic function $\wp$ satisfies the integral formula 
\begin{equation}
    z \equiv \int^{\infty}_{\wp(z;g_2,g_3)} \frac{dw}{\sqrt{4w^3 - g_2w -g_3}};
\end{equation}
the derivative of $\wp$ satisfies the relation 
\begin{equation}
    \left[ \frac{d \wp(z;g_2,g_3)}{ dz} \right]^2 =  4 \wp(z;g_2,g_3)^3 - g_2 \wp(z;g_2,g_3) - g_3.
\end{equation}
In what follows, we will use an abbreviated notation: $\wp(z)=\wp(z;g_2,g_3)$,  $\wp'(z)= d \wp(z;g_2,g_3)/ dz$.

\begin{lemma}[Euler 1761]
\label{lemma1a}
Let $f(x) = a_0 x^4 + 4 a_1 x^3 +6 a_2 x^2 + 4a_3 x + a_4$. The differential equation
\begin{equation}
\frac{dy}{dx} = \frac{\sqrt{f(y(x))}}{\sqrt{f(x)}}
\end{equation}
has an integral of the form
\begin{equation}
\label{eulerlemmaintegral}
    \left[ \frac{\sqrt{f(x)} + \sqrt{f(y(x))}}{x-y(x)} \right]^2 = a_0 [x+y(x)]^2 + 4a_1 [x+y(x)] + w^\prime,
\end{equation}
where $w^\prime$ is an integration constant. Similarly, equation
\begin{equation}
\frac{dy}{dx} = - \frac{\sqrt{f(y(x))}}{\sqrt{f(x)}}
\end{equation}
has an integral of the form
\begin{equation}
\label{eulerlemmaintegralminus}
    \left[ \frac{\sqrt{f(x)} - \sqrt{f(y(x))}}{x-y(x)} \right]^2 = a_0 [x+y(x)]^2 + 4a_1 [x+y(x)] + w^\prime,
\end{equation}
\end{lemma}

\begin{proof}
The following proof is due to Lagrange \cite{greenhill_applications_1892}. For simplicity, we only give the poof of Eq.\ (\ref{eulerlemmaintegral}). Equation (\ref{eulerlemmaintegralminus}) can be proved in an analogous way.

Let us introduce a new independent variable $\Lambda$ such that $dx/d\Lambda = \sqrt{f(x)}$. With a slight abuse of notation we write $y(\Lambda) = y(x(\Lambda))$. It follows that $dy/d\Lambda = \sqrt{f(y(\Lambda))}$. Define $p(\Lambda) = x(\Lambda) + y(\Lambda)$ and $q(\Lambda) = x(\Lambda) - y(\Lambda)$, so that
\begin{equation}
\frac{dp}{d\Lambda} = \sqrt{f(x(\Lambda))} + \sqrt{f(y(\Lambda))}, \quad \frac{dq}{d\Lambda} = \sqrt{f(x(\Lambda))} - \sqrt{f(y(\Lambda))}.
\end{equation}
Differentiating further with respect to $\Lambda$, one gets
\begin{equation}
\frac{d^2 p}{d\Lambda^2} = \frac{1}{2} \left[ f^\prime(x(\Lambda)) + f^\prime(y(\Lambda)) \right] = \frac{1}{2} a_0 \left( p^3 + 3pq^2 \right) + 3a_1 \left( p^2 +q^2 \right) + 6a_2 p +4a_3
\end{equation}
and
\begin{equation}
    \frac{dp}{d\Lambda} \frac{dq}{d\Lambda} = f(x(\Lambda)) - f(y(\Lambda)) = \frac{1}{2} a_0 pq \left( p^2 + q^2 \right) + a_1 q \left( 3 p^2 + q^2 \right) + 6a_2 pq + 4a_3 q.
\end{equation}
Hence
\begin{equation}
    \frac{2}{q^2} \frac{dp}{d\Lambda} \frac{d^2 p}{d\Lambda^2}  -  \frac{2}{q^3} \frac{dq}{d\Lambda} \left( \frac{dp}{d\Lambda} \right)^2 = 2 a_0 p \frac{dp}{d\Lambda} + 4 a_1 \frac{dp}{d\Lambda}.
\end{equation}
The above equation can be readily integrated, yielding
\begin{equation}
    \left( \frac{1}{q} \frac{dp}{d\Lambda} \right)^2 = a_0 p^2 + 4 a_1 p + w
\end{equation}
or, equivalently,
\begin{equation}
    \left[ \frac{\sqrt{f(x)} + \sqrt{f(y(x))}}{x-y(x)} \right]^2 = a_0 (x+y)^2 + 4a_1 (x+y) + w^\prime,
\end{equation}
where $w^\prime$ is an integration constant.
\end{proof}

\begin{theorem}[Biermann-Weierstrass]
\label{Biermann-Weierstrass}
Let 
\begin{equation}
    f(x) = a_0 x^4 + 4 a_1 x^3 +6 a_2 x^2 + 4a_3 x + a_4,
\end{equation}
be a quartic polynomial. Denote the Weierstrass invariants of $f$ by $g_2$ and $g_3$, i.e.,
\begin{subequations}
\label{invariants_theorem}
\begin{eqnarray}
        g_2 & \equiv &  a_0 a_4 - 4a_1 a_3 + 3 a_2^2, \\
        g_3 & \equiv & a_0 a_2 a_4 + 2a_1 a_2 a_3 -a_2^3 -a_0 a_3^2 - a_1^2 a_4.
\end{eqnarray}
\end{subequations}
Let
\begin{equation}
\label{zthm}
     z(x) = \int^x_{x_0} \frac{dx^\prime}{\sqrt{f(x^\prime)}},
\end{equation}
where $x_0$ is any constant, not necessarily a zero of $f(x)$. Then 
\begin{equation}
\label{glowne}
     x = x_0 + \frac{- \sqrt{f(x_0)} \wp'(z) + \frac{1}{2} f'(x_0) \left( \wp(z) - \frac{1}{24}f''(x_0) \right) + \frac{1}{24} f(x_0) f'''(x_0)  }{2 \left( \wp(z) - \frac{1}{24} f''(x_0) \right)^2 - \frac{1}{48} f(x_0) f^{(4)}(x_0) },
\end{equation}
 and
\begin{subequations}
\label{BW_wp}
    \begin{equation}
        \wp(z)  =  \frac{\sqrt{f(x) f(x_0)} + f(x_0)}{2(x-x_0)^2} + \frac{f'(x_0)}{4(x-x_0)} + \frac{f''(x_0)}{24},\\
    \end{equation}
    \begin{equation}
        \wp'(z)  =  \textstyle  - \left[ \frac{f(x)}{(x-x_0)^3} -\frac{f'(x)}{4(x-x_0)^2} \right] \sqrt{f(x_0)} - \left[ \frac{f(x_0)}{(x-x_0)^3} + \frac{f'(x_0)}{4(x-x_0)^2} \right]\sqrt{f(x)},
    \end{equation}    
\end{subequations}
where $\wp(z) =\wp(z;g_2,g_3)$ is the Weierstrass function corresponding to invariants (\ref{invariants_theorem}).
\end{theorem}

\begin{proof}
In what follows, we assume that $z$, $x$, and $x_0$ are real. We also assume that $f(x) \ge 0$ in the interval $(x_0,x)$. Hence $z > 0$ for $x > x_0$, and conversely $z < 0$ for $x < x_0$. In the first step of the proof, we show that the integral (\ref{zthm}) can be transformed to the Weierstrass form, i.e., there exists a transformation $w = w(x)$ such that
\begin{equation}
\label{whatwewant}
    z(x) = \int^{x}_{x_0} \frac{d x^\prime}{\sqrt{f(x^\prime)}} = \pm \int_{w(x_0)}^{w(x)} \frac{dw^\prime}{\sqrt{4 {w^\prime}^3 - g_2 w^\prime - g_3}}.
\end{equation}
Quite remarkably, such a transformation is related to formula (\ref{eulerlemmaintegral}) of Lemma \ref{lemma1a}. Let us take
\begin{equation}
\label{wformula1}
    w(x) =  \frac{1}{4} \left[ \frac{\sqrt{f(x)} + \sqrt{f(y)}}{x-y} \right]^2 - \frac{1}{4} a_0 (x+y)^2 - a_1 (x+y) - a_2,
\end{equation}
where $y$ is treated as a paramter. Note that $w = \frac{1}{4} w^\prime - a_2$, where $w^\prime$ is the constant appearing in Eq.\ (\ref{eulerlemmaintegral}). A straightforward computation yields
\begin{equation}
\label{dwdx}
\frac{dw}{dx} = - \frac{A(x,y)}{\sqrt{f(x)}},
\end{equation}
where
\begin{equation}
\label{aduze}
        A(x,y) =\left[  \frac{f(x)}{(x-y)^3} -\frac{f'(x)}{4(x-y)^2} \right] \sqrt{f(y)} + \left[ \frac{f(y)}{(x-y)^3} + \frac{f'(y)}{4(x-y)^2} \right]  \sqrt{f(x)}.
\end{equation}
In the following, we restrict ourselves to the range in which $dw/dx$ [and hence also $A(x,y)$] has a definite sign, so that the map $x \to w(x)$ constitutes a valid change of the integration variable. One can show that $A^2(x,y) = W(w(x))$, where $W = 4 w^3 - g_2 w - g_3$, and the invariants $g_2$ and $g_3$ are given by Eqs.\ (\ref{invariants_theorem}). Consequently, $A(x,y) = \epsilon \sqrt{W(w(x))}$, where $\epsilon = \pm 1$. This proves Eq.\ (\ref{whatwewant}). More precisely,
\begin{equation}
    z(x) = \int^{x}_{x_0} \frac{d x^\prime}{\sqrt{f(x^\prime)}} = - \epsilon \int_{w(x_0)}^{w(x)} \frac{ dw^\prime}{\sqrt{4 {w^\prime}^3 - g_2 w^\prime - g_3}} = \epsilon \int^{w(x_0)}_{w(x)} \frac{ dw^\prime}{\sqrt{4 {w^\prime}^3 - g_2 w^\prime - g_3}}.
\end{equation}
Setting $y = x_0$ in Eq.\ (\ref{wformula1}), we get $w \to + \infty$ for $x \to x_0$. As a consequence, one obtains
\begin{equation}
\label{zweierstrass}
    z(x) = \int^{x}_{x_0} \frac{d x^\prime}{\sqrt{f(x^\prime)}} = \epsilon \int^{\infty}_{w(x)} \frac{ dw^\prime}{\sqrt{4 {w^\prime}^3 - g_2 w^\prime - g_3}}.
\end{equation}
Note that $\epsilon = +1$ for $x > x_0$, and $\epsilon = - 1$ for $x < x_0$. It follows that $w(x)$ can be written as $w(x) = \wp(\epsilon z(x); g_2,g_3) = \wp(z(x); g_2,g_3)$ and $\sqrt{W(w(x))} = - \epsilon  \wp^\prime(z(x);g_2,g_3)$ (the last relation can be obtained directly by differentiating Eq.\ (\ref{zweierstrass}) with respect to $x$).

It is easy to check that $w$ defined by Eq.\ (\ref{wformula1}) can be also written as
\begin{equation}
\label{walternative}
    w = \frac{F_1(x,y) + \sqrt{f(x) f(y)}}{2(x-y)^2},
\end{equation}
where $F_1(x,y)= f(y) + \frac{1}{2} f'(y) (x-y) +\frac{1}{12} f''(y) (x-y)^2$. It is a positive root of the quadratic equation
\begin{equation}
\label{proof10}
    (x-y)^2 w^2 - F_1(x,y) w +F_2(x,y) =0,
\end{equation}
where
\begin{equation}
F_2(x,y) = \frac{F_1(x,y)^2 - f(x)f(y)}{4(x-y)^2}.
\end{equation}
Using the relation 
\begin{equation}
     f(x) = f(y) + f'(y) (x-y) + \frac{1}{2} f''(y) (x-y)^2 + \frac{1}{6} f'''(y)(x-y)^3 + \frac{1}{24} f^{(4)}(y) (x-y)^4,
\end{equation}
one can transform Eq.\ (\ref{proof10}) into the form 
\begin{eqnarray}
\left[ w^2 - \frac{1}{12} w f''(y) + \frac{1}{576} f''(y)^2 - \frac{1}{96} f(y)f^{4}(y) \right] (x-y)^2 &&  \nonumber\\
+ \left[ -\frac{1}{2} w f'(y) + \frac{1}{48} f'(y) f''(y) - \frac{1}{24} f(y) f'''(y) \right] (x-y) && \nonumber \\
- w f(y) + \frac{1}{16} f'(y)^2 - \frac{1}{12} f(y) f''(y)  & = & 0,
\end{eqnarray}
which is a quadratic equation with respect to $x - y$. Solutions of this equation can be written as
\begin{equation}
    x-y = \frac{\pm \sqrt{f(y)} \sqrt{W} + \frac{1}{2} f'(y) \left[ w - \frac{1}{24} f''(y) \right] + \frac{1}{24} f(y) f'''(y)}{2\left[ w - \frac{1}{24} f''(y) \right]^2 -\frac{1}{48} f(y) f^{(4)}(y) }.
\end{equation}
A close inspection shows that the plus and minus sign in the above expression is correlated with the sign of $x - y$. We have
\begin{equation}
    x-y = \frac{+ \sqrt{f(y)} \sqrt{W} + \frac{1}{2} f'(y) \left[ w - \frac{1}{24} f''(y) \right] + \frac{1}{24} f(y) f'''(y)}{2\left[ w - \frac{1}{24} f''(y) \right]^2 -\frac{1}{48} f(y) f^{(4)}(y) }
\end{equation}
for $x > y$, and
\begin{equation}
    x-y = \frac{- \sqrt{f(y)} \sqrt{W} + \frac{1}{2} f'(y) \left[ w - \frac{1}{24} f''(y) \right] + \frac{1}{24} f(y) f'''(y)}{2\left[ w - \frac{1}{24} f''(y) \right]^2 -\frac{1}{48} f(y) f^{(4)}(y) }
\end{equation}
for $x < y$. This observation follows from noticing that $w(x)$ is a decreasing function of $x$ for $x > y$ and an increasing function of $x$ for $x < y$ [cf.\ Eq.\ (\ref{dwdx})], and from inspecting the limits of the above expressions for $w \to + \infty$. Given that $\wp^\prime(z(x);g_2,g_3) = - \epsilon \sqrt{W(w(x))} = - A(x,x_0)$, and returning to our choice $y = x_0$, we write the expression for $x$ as
\begin{equation}
    x = x_0 + \frac{- \sqrt{f(x_0)} \wp'(z) + \frac{1}{2} f'(x_0) \left[ \wp(z) - \frac{1}{24}f''(x_0) \right] + \frac{1}{24} f(x_0) f'''(x_0)  }{2 \left[ \wp(z) - \frac{1}{24} f''(x_0) \right]^2 - \frac{1}{48} f(x_0) f^{(4)}(x_0)},
\end{equation}
i.e., in the form of Eq.\ (\ref{glowne}). Equations (\ref{BW_wp}) follow directly from Eqs.\ (\ref{aduze}) and (\ref{walternative}).
\end{proof}

\section{Classification of trajectories}
\label{sec:class_of_traj}

\subsection{Timelike trajectories}

Qualitative behavior of timelike trajectories depend on the properties of the dimensionless effective radial potential (\ref{eff_pot}). The motion of a massive particle is only possible in regions where
\begin{equation}
\label{mov_cond}
    \varepsilon^{2} - U_\lambda(\xi) \geq 0.
\end{equation}
For $0 \leq \lambda^2 \leq 12$, $U_\lambda(\xi)$ is a monotonically increasing function of $\xi$, growing from 0 at $\xi = 2$ to $U_\lambda(\xi) \to 1$ for $\xi \to \infty$. For $\lambda^2 > 12$, $U_\lambda(\xi)$ has two local extrema: a local minimum at  
\begin{equation}
    \xi_\mathrm{min} = \frac{\lambda^2}{2} \left( 1 + \sqrt{1 - \frac{12}{\lambda^2} } \right),
\end{equation}
and a local maximum at
\begin{equation}
\label{ximax}
    \xi_\mathrm{max} = \frac{\lambda^2}{2} \left( 1 - \sqrt{1 - \frac{12}{\lambda^2} } \right)
\end{equation}
(see e.g.\ \cite{rioseco_accretion_2017}). The location of the minimum $\xi_\mathrm{min}$ grows monotonically from 6 to infinity, as $\lambda^2$ increases from 12 to infinity. At the same time, the radius $\xi_\mathrm{max}$ decreases monotonically from 6 to 3. We have
\begin{equation}
    U_\lambda(\xi_\mathrm{min}) = \frac{8}{9} + \frac{\lambda^2 - 12}{9 \xi_\mathrm{min}}, \quad U_\lambda(\xi_\mathrm{max}) = \frac{8}{9} + \frac{\lambda^2 - 12}{9 \xi_\mathrm{max}}.
\end{equation}
The value $U_\lambda(\xi_\mathrm{min})$ grows from $8/9$ to 1, as $\lambda^2$ increases from $12$ to infinity. Simultaneously, the value $U_\lambda(\xi_\mathrm{max})$ grows from $8/9$ to infinity. For $\lambda^2 \ge 16$, the value of the potential at the maximum is always greater than or equal to one; otherwise, it is smaller.

Consider an equation
\begin{equation}
\label{lambdaceq}
    U_\lambda (\xi_\mathrm{max}) = \varepsilon^2,
\end{equation}
i.e., a limiting case of inequality (\ref{mov_cond}), where $\xi_\mathrm{max}$ is given by Eq.\ (\ref{ximax}). A solution of Eq.\ (\ref{lambdaceq}) with respect to $\lambda^2$ reads
\begin{equation}
        \label{lambda-crit-unbound}        
            \lambda_\mathrm{c}(\varepsilon)^2 = \frac{12}{1 - \frac{4}{\left( \frac{3 \varepsilon}{\sqrt{9\varepsilon^2 -8}}+1 \right)^2}}.
\end{equation}
In other words, $\lambda_\mathrm{c}(\varepsilon)$ denotes the value of the angular momentum for which the radial potential at the local maximum is equal to $\varepsilon^2$. It turns out to be particularly useful in classifying different types of timelike trajectories.

There are three main types of orbits: radial, bound, and unbound. They can be characterized as follows.
\begin{itemize}
    \item Type I (\textit{radial orbits}). This class consists of trajectories for which $\lambda^2=0$. Test particles move radially.
    \item Type II (\textit{bound orbits}). Bound orbits never reach $\xi = \infty$. They can be divided into the following sub-types:
    \begin{itemize}
        \item[a)] \textit{Inner orbits}. This is a class of bound orbits with at least one of the endpoints beneath the black hole horizon. For $\lambda^2 < 12$, this is the only type of bound orbits. The energy associated with such orbits is limited by $\varepsilon^2 < 1$ for $\lambda^2 < 16$ and $\varepsilon^2 \le U_\lambda(\xi_\mathrm{max})$ for $16 \le \lambda^2$. For $\lambda^2 \ge 12$, there are limiting cases with $\lambda = \lambda_\mathrm{c}(\varepsilon)$, in which the orbits can spiral asymptotically towards $\xi = \xi_\mathrm{max}$.
        \item[b)] \textit{Outer orbits}. These are trajectories trapped in a potential well, which can exist for $3 < \xi_\mathrm{max} \le \xi$. In the generic case of outer bound orbits, the equation $\varepsilon^2 = U_\lambda(\xi)$ has 3 positive roots $\xi_1$, $\xi_2$, $\xi_3$, satisfying $\xi_1 < \xi_\mathrm{max} < \xi_2 < \xi_\mathrm{min} < \xi_3$, and the particle oscillates between $\xi_2$ and $\xi_3$. Thus outer bound orbits exist for $12 \le \lambda^2$. For $12 \le \lambda^2 < 16$, the energy $\varepsilon$ is bounded by $U_\lambda(\xi_\mathrm{min}) \le \varepsilon^2 \le U_\lambda(\xi_\mathrm{max}) < 1$. For $16 \le \lambda^2$, the energy $\varepsilon$ satisfies $U_\lambda(\xi_\mathrm{min}) \le \varepsilon^2 < 1$.
        
        Alternatively, the phase space occupied by outer bound orbits can be characterized by
        \begin{equation}
            \tilde \varepsilon_\mathrm{min} \le \varepsilon < 1, \quad \lambda_\mathrm{c}(\varepsilon) \le \lambda \le \lambda_\mathrm{max}(\varepsilon,\xi),
        \end{equation}
        where
        \begin{equation}
        \label{varepsilon-min-bound}
            \tilde \varepsilon_\mathrm{min} = 
            \begin{cases}
                \infty & \xi \leq 3, \\
                \sqrt{\left( 1-\frac{2}{\xi} \right) \left( 1+\frac{1}{\xi-3} \right)} & 3 < \xi,
            \end{cases}
        \end{equation}
        and
        \begin{equation}
        \label{lambda-max}        
            \lambda_\mathrm{max}(\varepsilon,\xi) = \xi\sqrt{\frac{ \varepsilon^2}{1-\frac{2}{\xi}} -1}
    \end{equation}
    (see, e.g., \cite{Olivier_disks}). Note that circular orbits with either $\xi = \xi_\mathrm{max}$ (stable) or $\xi = \xi_\mathrm{min}$ (unstable) belong to this class. There are also limiting cases with $\lambda = \lambda_\mathrm{c}(\varepsilon)$, similar to the limiting cases of Type IIa and Type IIIc.
    \end{itemize}
    \item Type III (\textit{unbound orbits}). In this case $\varepsilon \ge 1$. Unbound trajectories are divided into the following two sub-types: 
    \begin{itemize}
        \item[a)] \textit{Absorbed orbits}. These trajectories originate at $\xi = \infty$ and end beneath the black hole horizon. The angular momentum associated with absorbed trajectories satisfies $\lambda < \lambda_\mathrm{c}(\varepsilon)$.
        \item[b)] \textit{Scattered orbits}. Both endpoints of scattered trajectories reach infinity. The particles never reach below $\xi = 3$, i.e., below the photon sphere. Their energy is bounded from below by
        \begin{equation}
        \label{varepsilon-min-unbound}
            \varepsilon_\mathrm{min} = 
            \begin{cases}
                \infty & \xi \leq 3, \\
                \sqrt{\left( 1-\frac{2}{\xi} \right) \left( 1+\frac{1}{\xi-3} \right)} & 3<\xi < 4, \\
                1 & \xi \geq 4.
            \end{cases}
        \end{equation}
        The total angular momentum of a  scattered particle is limited from above, i.e., $\lambda_\mathrm{c}(\varepsilon) < \lambda \leq  \lambda_\mathrm{max}(\varepsilon,\xi)$.
        \item[c)] A limiting case with $\lambda = \lambda_\mathrm{c}(\varepsilon)$. The particle travels from infinity and spirals asymptotically to $\xi = \xi_\mathrm{max}$.
    \end{itemize}
\end{itemize}

A comprehensive discussion of the classification of orbits can be found in \cite{rioseco_accretion_2017,chandrasekhar_mathematical_1983,kostic_analytical_2012}. Figure \ref{fig:eff_pot_4_2} shows the radial effective potential $U_\lambda (\xi)$ corresponding to different types of orbits listed above.

\begin{figure}[t]
\centering
\includegraphics[width=0.45\linewidth]{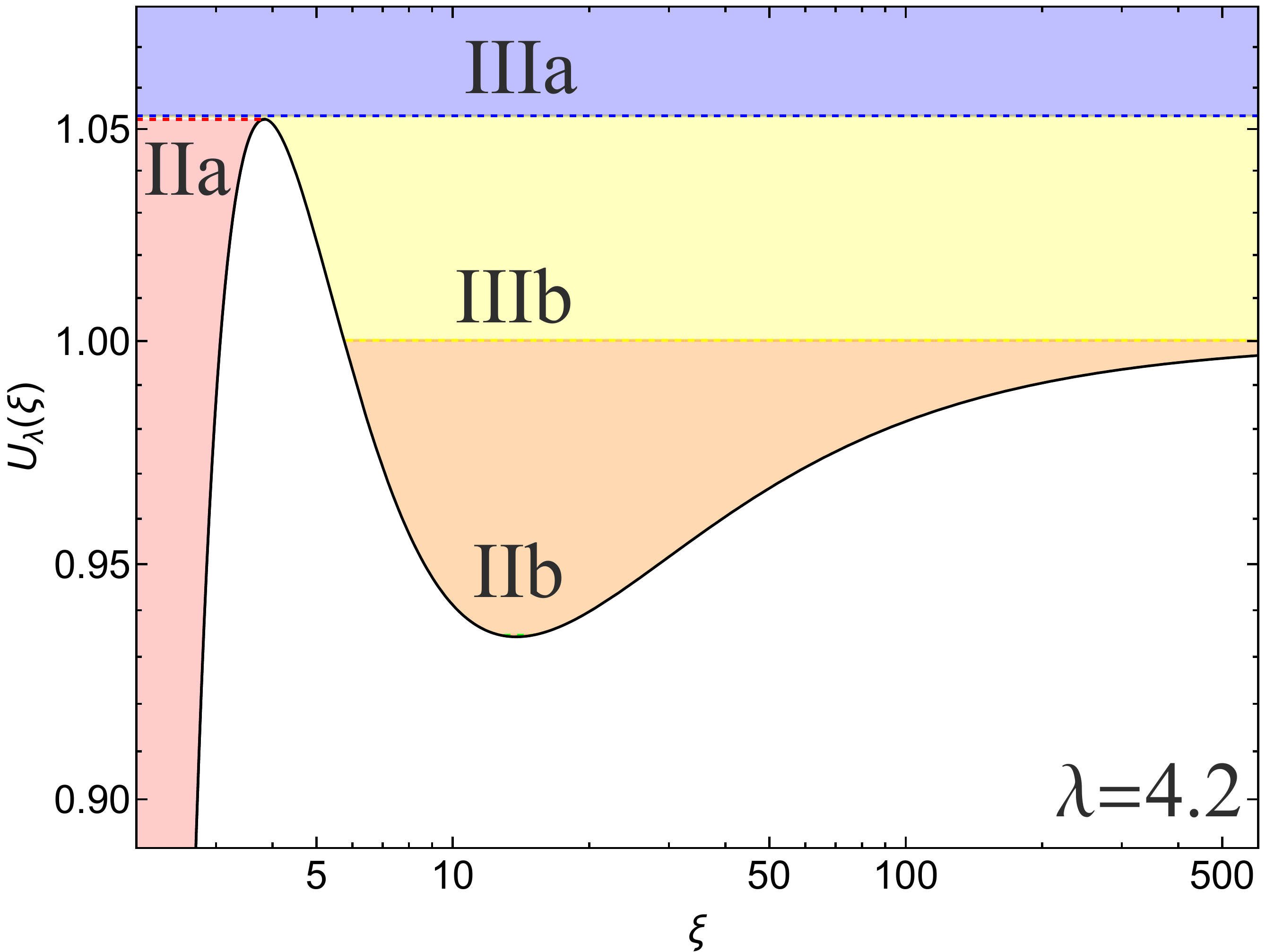}
\includegraphics[width=0.45\linewidth]{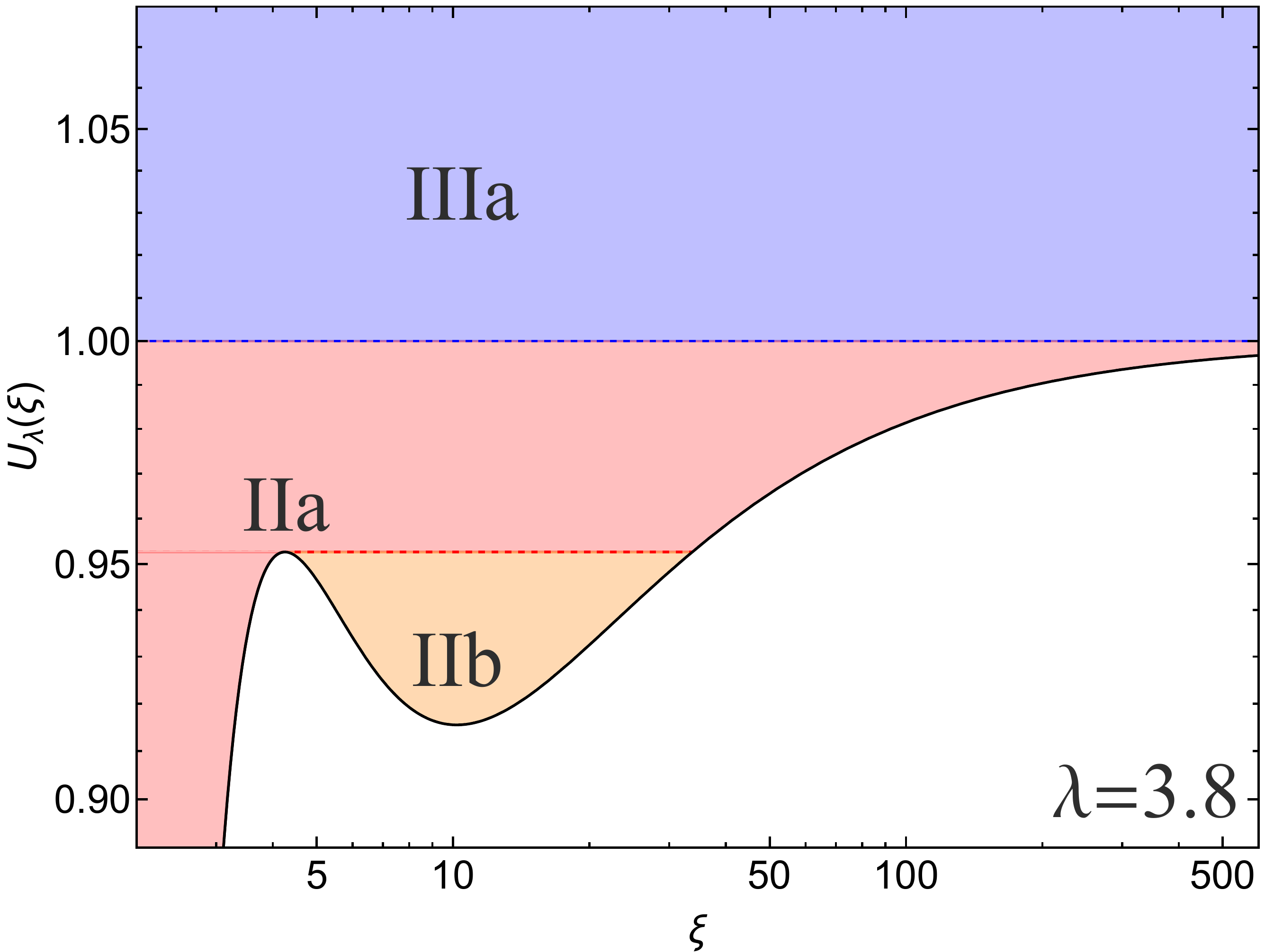}
\caption{\label{fig:eff_pot_4_2}
The effective potential $U_{\lambda}(\xi)$ [Eq.\ (\ref{eff_pot})] for $\lambda=4.2$ (left) and  $\lambda=3.8$ (right). The properties of the orbit depend on the energy of the particle and the location with respect to the local maximum of $U_\lambda(\xi)$. Different types of orbits (IIa, IIb, IIIa, IIIb) are marked with separate colors.}
\end{figure}

\subsection{Null trajectories}

The classification of null trajectories is similar to that of timelike orbits, but it is in many respects much simpler. The dimensionless radial potential $U_\lambda(\xi)$, defined by Eq.\ (\ref{eff_pot_null}), has a single maximum at $\xi_\mathrm{max} = 3$. The value of the potential at the maximum reads
\begin{equation}
    U_\lambda(\xi_\mathrm{max}) = \frac{\lambda^2}{27}.
\end{equation}
Consequently, the equivalent of the function $\lambda_c(\varepsilon)$, defined in Eq.\ (\ref{lambdaceq}), reads
\begin{equation}
    \lambda_c(\varepsilon) = \sqrt{27} \varepsilon.
\end{equation}

The orbits are divided into the following classes.

\begin{itemize}
    \item Type I (\textit{radial orbits}). As for timelike geodesics, this class consists of orbits with $\lambda = 0$.
    \item Type II (\textit{bound orbits}). Null bound orbits exist for $\xi \le 3$ and $\varepsilon^2 \le \lambda^2/27$. This type includes (as a limiting case) the circular photon orbit with the radius $\xi = 3$.
    \item Type III (\textit{unbound orbits}). As for timelike trajectories, unbound orbits can be divided into to following types.
    \begin{enumerate}
        \item[a)] \textit{Absorbed orbits.} In this case $\varepsilon^2 > 0$ and $\lambda < \lambda_c(\varepsilon)$.
        \item[b)] \textit{Scattered orbits.} For scattered orbits $\lambda_c(\varepsilon) < \lambda \le \lambda_\mathrm{max}(\varepsilon,\xi)$, where
        \begin{equation}
            \lambda_\mathrm{max}(\varepsilon,\xi) = \frac{\varepsilon \xi}{\sqrt{1 - \frac{2}{\xi}}}.
        \end{equation}
        These orbits exist only for $\xi > 3$.
        \item[c)] A limiting case with $\lambda = \lambda_c(\varepsilon)$.
    \end{enumerate}
\end{itemize}

\section{Elliptic expressions for $\psi(\xi)$}
\label{appendix:psi_xi}

In this appendix we derive expressions for
\begin{equation}
    X(\xi_0) = \int_{\xi_0}^{\infty} \frac{d \xi}{\sqrt{ f(\xi)}}
\end{equation}
and
\begin{equation}
    Y(\xi_0;\tilde \xi) = \int_{\xi_0}^{\tilde \xi} \frac{d \xi}{\sqrt{ f(\xi)}}
\end{equation}
in terms of Legendre elliptic integrals. Of course, $X(\xi_0) = Y(\xi_0;\infty)$. Substitutions given in this appendix are known, and they are used e.g.\ in \cite{kostic_analytical_2012}. They can be introduced quite generally, both for timelike and null orbits.

Let us start with a more general elliptic integral
\begin{equation}
    I = \int_a^b \frac{d\xi}{\sqrt{f(\xi)}},
\end{equation}
where $f(\xi) =  a_0 \xi^4 + 4 a_1 \xi^3 +6 a_2 \xi^2 + 4a_3 \xi + a_4$. Let $e$ be a zero of $f$. Substituting $\xi = e + \frac{1}{x}$, we get 
\begin{equation}
\label{Ix}
    I = -  \int_\frac{1}{a-e}^\frac{1}{b-e} \frac{dx}{\sqrt{A_0 + 4A_1 x + 6A_2 x^2 + 4A_3 x^3}},
\end{equation} 
where $A_0 = a_0$, $A_1 = a_1 + a_0 e$, $A_2 = a_2 + 2 a_1 e +a_0 e^2$, $A_3 = a_3 + 3 a_2 e +3 a_1 e^2 + a_0 e^3$. The transformation $\xi = e + \frac{1}{x}$ maps the zero $\xi = e$ to infinity, removing one factor $(\xi - e)$ from the fectorization of $f(\xi)$. The remaining zeros of $f(\xi)$ are mapped into zeros of $A_0 + 4A_1 x + 6A_2 x^2 + 4A_3 x^3$.  Next, another substitution $x = \frac{y-\frac{1}{2} A_2}{A_3}$, brings the above integral to the Weierstrass form 
\begin{equation}
    I = - \mathrm{sgn}(A_3) \int_{\frac{A_2}{2}+\frac{A_3}{a-e}}^{\frac{A_2}{2}+\frac{A_3}{b-e}} \frac{dy}{\sqrt{4 y^3 - g_2 y - g_3}},
\end{equation} 
where $g_2$ and $g_3$ are given by (\ref{invariants_theorem}).

Further reduction to Legendre integrals requires a control of the integration range with respect to zeros of the polynomial $4y^3 - g_2 y - g_3$, provided that one wants to have explicitly real expressions. We start by computing the integral $X(\xi_0)$, assuming a scattered unbound orbit. In this case, $f(\xi)$ has four real zeros, one of which is simply $\xi = 0$. This can be seen as follows. We have $f(\xi) = \xi^4 (\varepsilon^2 - U_\lambda(\xi))$. The expression $\varepsilon^2 - U_\lambda(\xi)$ has three real zeros: two of them are positive, as follows from the discussion concerning the centrifugal barrier. The third one is negative. For timelike trajectories this fact can be seen by noting that
\begin{equation}
    \varepsilon^2 - U_\lambda(\xi) = \varepsilon^2 - 1 + \frac{1}{\xi^3} \left( 2 \xi^2 - \lambda^2 \xi + 2 \lambda^2 \right).
\end{equation}
Conesequently, $\varepsilon^2 - U_\lambda(\xi)$ tends to $-\infty$ for $\xi \to 0_-$, and $\varepsilon^2 - U_\lambda(\xi) \to \varepsilon^2 - 1 > 0$, for $\xi \to - \infty$. Hence, $\varepsilon^2 - U_\lambda(\xi)$ changes its sign for $\xi < 0$. For null trajectories the reasoning is analogous, but this time
\begin{equation}
    \varepsilon^2 - U_\lambda(\xi) = \varepsilon^2 - \left( 1 - \frac{2}{\xi} \right) \frac{\lambda^2}{\xi^2}.
\end{equation}
Consequently, $\varepsilon^2 - U_\lambda(\xi) \to \varepsilon^2 > 0$, for $\xi \to - \infty$. In both cases (timelike and null) we will denote the zeros of $f(\xi)$ as $\xi_1 < 0 < \xi_3 < \xi_2$. 

In the substitution $\xi = e + \frac{1}{x}$ leading to Eq.\ (\ref{Ix}), we now choose $e = 0$. This yields $A_0 =a_0$, $A_1 = a_1$, $A_2 = a_2$, $A_3 = a_3$. Note that $a_2 = -1/6$, $a_3 = 1/2$, both for timelike and null trajectories. Hence, $X(\xi_0)$ can be written as
\begin{equation}
\label{Xxi0}
    X(\xi_0) = \int^{-\frac{1}{12}+\frac{1}{2 \xi_0}}_{-\frac{1}{12}} \frac{dy}{\sqrt{4 y^3 - g_2 y - g_3}}.
\end{equation}
The substitution $y = -\frac{1}{12}+\frac{1}{2 \xi}$ maps the zeros $\xi_1$, $\xi_2$, $\xi_3$ of $f(\xi)$ to $y_1$, $y_2$, $y_3$, respectively, but this time $y_1 < y_2 < y_3$. The original integration range of $\xi$, $\xi_2 \le \xi_0 \le \xi < \infty$, is mapped into the segment: $y_1 < -1/12 < y \le y_2$.

We now make a substitution
\begin{equation}
\label{subs3}
    y = y_2 - \mu^2 \cos^2 \chi, \quad \mu^2 = y_2 - y_1, \quad k^2 = \frac{y_2 - y_1}{y_3 - y_1}, \quad 0 \le \chi \le \frac{\pi}{2}.
\end{equation}
Thus $y = y_1$ for $\chi = 0$, and $y = y_2$ for $\chi = \pi/2$. This yields
\begin{equation}
    \int \frac{dy}{\sqrt{4 (y-y_1)(y-y_2)(y-y_3)}} = \frac{k}{\mu} \int \frac{d \chi}{\sqrt{1 - k^2 \sin^2 \chi}},
\end{equation}
where $k/\mu = 1/\sqrt{y_3 - y_1}$. Consequently,
\begin{equation}
    X(\xi_0) = \frac{1}{\sqrt{y_3 - y_1}} \left[ F \left( \arccos \sqrt{ \frac{y_2 + \frac{1}{12} - \frac{1}{2\xi_0}}{y_2 - y_1}} , k \right) - F \left( \arccos \sqrt{\frac{y_2 + \frac{1}{12}}{y_2 - y_1}}, k \right) \right],
\end{equation}
where
\begin{equation}
    F(\phi,k) = \int_0^\phi \frac{d \chi}{\sqrt{1 - k^2 \sin^2 \chi}}, \quad -\frac{\pi}{2} < \phi < \frac{\pi}{2}.
\end{equation}

For generic unbound absorbed orbits, the situation is slightly different. In this case the polynomial $4y^3 - g_2y - g_3$ has only one real zero $y_1 < -1/12$. We write: $4y^3 - g_2y - g_3 = 4(y - y_1)(y^2 + py + q)$, where $p^2 - 4 q < 0$ and thus $y^2 +p y + q > 0$. The substitution which turns Eq.\ (\ref{Xxi0}) into the Legendre form reads now
\begin{equation}
\label{subs1}
    y = y_1 + \mu \tan^2 \frac{\chi}{2}, \quad \mu = \sqrt{y_1^2 + p y_1 + q}, \quad 0 \le \chi < \frac{\pi}{2}.
\end{equation}
We have $y = y_1$ for $\chi = 0$ and $y \to \infty$ for $\chi \to \pi/2$. A straightforward algebraic calculation yields now
\begin{equation}
    \int \frac{dy}{\sqrt{4 (y - y_1)(y^2 + p y + q)}} = \frac{1}{2\sqrt{\mu}} \int \frac{d \chi}{\sqrt{1 - k^2 \sin^2 \chi}},
\end{equation}
where
\begin{equation}
\label{k2bis}
    k^2 = \frac{1}{2} \left( 1 - \frac{y_1 + p/2}{\mu} \right).
\end{equation}
Note that $k^2$ is real and positive. This fact follows directly from the inequality $p^2 - 4 q < 0$, which ensures that
\begin{equation}
    \left( \frac{y_1 + p/2}{\mu} \right)^2 = \frac{y_1^2 + p y_1 + \frac{p^2}{4}}{y_1^2 + p y_1 + q} < 1.
\end{equation}
As a consequence, $X(\xi_0)$ can be written as
\begin{equation}
\label{xbis}
    X(\xi_0) = \frac{1}{2 \sqrt{\mu}} \left[ F\left(2 \arctan \sqrt{\frac{- \frac{1}{12} + \frac{1}{2 \xi_0} - y_1}{\mu}} , k \right) - F \left( 2 \arctan \sqrt{\frac{-\frac{1}{12} - y_1}{\mu}} , k \right) \right].
\end{equation}

For generic timelike outer bound orbits the expression $\varepsilon^2 - U_\lambda(\xi)$ has 3 real positive roots $\xi_3 < \xi_2 < \xi_1$, which are also the zeros of $f(\xi)$ (the fourth root being $\xi_4 = 0$). The motion is allowed in the range $\xi_2 \le \xi \le \xi_1$. The transformation $y = -\frac{1}{12} + \frac{1}{2 \xi}$ maps the zeros $\xi_3$, $\xi_2$, $\xi_1$ into $y_1 < y_2 < y_3$ (the zeros of $4y^3 - g_2 y - g_3$). Explicitly real expressions for $Y(\xi_0,\tilde \xi)$ can be obtained with substitutions $(\ref{subs3})$. We get
\begin{equation}
     Y(\xi_0;\tilde \xi) = \frac{1}{\sqrt{y_3 - y_1}} \left[ F \left( \arccos \sqrt{ \frac{y_2 + \frac{1}{12} - \frac{1}{2\xi_0}}{y_2 - y_1}} , k \right) - F \left( \arccos \sqrt{\frac{y_2 + \frac{1}{12}-\frac{1}{2\tilde \xi}}{y_2 - y_1}}, k \right) \right],
\end{equation}

The case of inner bound orbits is more complex, since, depending on the values of $\varepsilon$ and $\lambda$, they correspond either to a case with three real zeros of $4y^3 - g_2 y - g_3$ or to a case in which this polynomial has just one real zero and two complex ones. Here again, substitutions (\ref{subs3}) and (\ref{subs1}) work, but one has to adjust the details (carefully select the roots $y_1$, $y_2$, and $y_3$).

\end{document}